\documentclass[12pt]{article}
\usepackage[utf8]{inputenc}
\usepackage[english]{babel}
 \usepackage{amsmath}
\usepackage{amsthm}
\usepackage{amssymb}
\usepackage{hyperref}
\usepackage{braket}
\usepackage{color}
\usepackage{comment}
\usepackage{graphicx}
\usepackage{stmaryrd}
   \usepackage[rightcaption]{sidecap}
   
   \newenvironment{assum} 
  {\list{}{\leftmargin=0.2in\rightmargin=0.2in}\item[]} 
  {\endlist}
 
    \newcommand{\Ab}[1]{\llbracket #1 \rrbracket}
   
\usepackage[margin=3cm]{geometry}

\renewcommand{\Gamma}{\mathbb{Z}^d}
 \newcommand{\epsi}{\varepsilon}
 \newcommand{\E}{{\mathrm{e}}}
 \newcommand{\I}{\mathrm{i}}
 \newcommand{\R}{ \mathbb{R} }
 \newcommand{\C}{ \mathbb{C} }
 \newcommand{\N}{ \mathbb{N} }
 \newcommand{\Z}{ \mathbb{Z} }
 \newcommand{\D}{\mathrm{d}}
 \newcommand{\Or}{{\mathcal{O}}}

 \newcommand{\bm}{\begin{pmatrix}}
 	\newcommand{\Em}{\end{pmatrix}}

 \newcommand{\Aloc}{\mathcal{A}_{\rm loc}}
 \newcommand{\SLT}{SLT}
\begin{document}
	\allowdisplaybreaks
\numberwithin{equation}{section}

	\theoremstyle{plain} 
	\newtheorem{prop}{Proposition}[section]
	\newtheorem{thm}[prop]{Theorem}
	\newtheorem{lem}[prop]{Lemma}
	\newtheorem{cor}[prop]{Corollary}
	
	 \theoremstyle{definition}
	\newtheorem{defi}[prop]{Definition}
	\newtheorem{rmk}[prop]{Remark}

\title{Adiabatic theorem in the thermodynamic limit:\\ Systems with a gap in the bulk.}

\author{ Joscha Henheik \\
	\textit{\footnotesize Mathematisches Institut, Eberhard-Karls-Universit\"at,} \\{ \it \footnotesize Auf der Morgenstelle 10, 72076
		T\"ubingen, Germany} \\
	{\footnotesize and} \\
	{\it \footnotesize IST Austria, Am Campus 1, 3400 Klosterneuburg, Austria} \\[7mm]
 Stefan Teufel\footnote{stefan.teufel@uni-tuebingen.de} \\ \textit{\footnotesize Mathematisches Institut, Eberhard-Karls-Universit\"at,} \\{ \it \footnotesize Auf der Morgenstelle 10, 72076
		T\"ubingen, Germany} }
\maketitle

\begin{abstract}
We prove a generalised super-adiabatic theorem for extended fermionic systems assuming a spectral gap only in the bulk. More precisely, we assume that the infinite system has a unique ground state and that the corresponding  GNS-Hamiltonian has a spectral gap above its eigenvalue zero. Moreover, we show that a similar adiabatic theorem also holds in the bulk of finite systems up to errors that vanish faster than any inverse power of  the system size,  although the corresponding finite volume Hamiltonians need not have a spectral gap. 
\end{abstract}

\section{Introduction}	

We prove an adiabatic theorem for 
 the automorphism group acting on the observable algebra  of quasi-local operators that  describes the dynamics of an interacting fermionic quantum system on the lattice~$\Z^d$ with short range interactions.
For this we assume that the generator of the automorphism group, the Liouvillian,  has a unique ground state and that the GNS Hamiltonian constructed from it has a spectral gap above zero.  

 Such systems are usually obtained by taking the thermodynamic limit $\Lambda\nearrow\Z^d$ for a sequence of systems defined on finite  domains $\Lambda\subset \Z^d$. In \cite{henheikteufel20202} we prove  an analogous  adiabatic theorem assuming a uniform gap  for  each finite system in such a sequence.  A spectral gap above the ground state, i.e.\ a minimal energy for excitations,  is characteristic for insulating materials. 
However, on finite domains it is often the case that one-body states supported close to the boundary, so called edge states, allow for excitations with arbitrarily small energy and thus the corresponding finite volume Hamiltonians have no  spectral gap.  
While in many models edge states can be suppressed by choosing appropriate boundary conditions, our result shows that this is not a necessary condition for the adiabatic theorem to hold in the thermodynamic limit. Moreover,  we  obtain an adiabatic theorem even for finite systems with edge states by ``restricting'' the infinite volume result to the bulk of the finite system.
 
While a   precise statement of our result requires considerable preparation, let us at least indicate a few more details here. Let $\{ H_0^\Lambda(t)\}_{\Lambda \subset\Z^d, \,t\in I}$ be a time-dependent family of Hamiltonians with local interactions describing the system on  finite domains~$\Lambda$. It is known that under quite general conditions (see e.g.\  \cite{bratteli1996operator, nachtergaele2019quasi}) the Heisenberg time-evolution generated by $H_0^\Lambda(t)$ at fixed $t\in I\subset \R$   induces in the thermodynamic limit $\Lambda\nearrow\Z^d$ a one-parameter group  $s\mapsto \E^{\I s \mathcal{L}_{H_0(t)}}$ of automorphisms on the observable algebra~$\mathcal{A}$ of quasi-local operators for the infinite system. Here $\mathcal{L}_{H_0 (t)} \Ab{\cdot} :=\lim_{\Lambda\nearrow\Z^d}  [H_0^\Lambda(t), \,\cdot\,]$
denotes the Liouvillian, a densely defined derivation on~$\mathcal{A}$. We assume that for each $t\in I$ the Liouvillian $ \mathcal{L}_{H_0(t)}$ has a unique ground state~$\rho_0(t)$. The most important innovation of our result is that we only assume that the ``infinite volume Hamiltonian'' associated with  $\rho_0(t)$  via the GNS construction has a spectral gap above its ground state, whereas in all previous work on adiabatic theorems for extended interacting systems\footnote{Note that the result of \cite{moon2019automorphic} is not an adiabatic theorem in the usual sense.} \cite{bachmann2018adiabatic,monaco2017adiabatic,teufel2020non,henheikteufel20202} the gap assumption was imposed on the finite volume Hamiltonians $H_0^\Lambda(t)$. 
 
Let $\mathfrak{U}_{t,t_0}^{ \eta}$ be the cocycle of automorphisms  on $\mathcal{A}$ generated by
the time-dependent Liouvillian $\frac{1}{\eta}\mathcal{L}_{H_0 (\cdot)} $,
 i.e.\ the physical time-evolution of the infinite system. The adiabatic parameter  $\eta>0$ controls the separation of two time-scales: the internal time-scale defined by the spectral gap and the time-scale on which the Hamiltonian varies, e.g.\ due to time-dependent external fields. The asymptotic regime $\eta\ll 1$ is called the adiabatic limit. 
The following infinite volume version of the standard super-adiabatic theorem is 
a special case of our result:
There exist super-adiabatic states~$\rho_0^\eta(t)$   on $\mathcal{A}$ that are close to the instantaneous ground states  $\rho_0(t)$ in the sense that
\[
|\rho_0^\eta(t)(A) - \rho_0(t)(A) | = \Or (\eta) \quad\mbox{for all } \,A\in\mathcal{D}\subset \mathcal{A}\,,
\]
such that the   evolution $\mathfrak{U}^{ \eta}_{t,t_0}$  intertwines the super-adiabatic states to all orders in the adiabatic parameter $\eta$, i.e., for all $A\in\mathcal{D}\subset \mathcal{A}$
\begin{equation}\label{standardAT}
|\rho_0^\eta(t_0)(\mathfrak{U}_{t,t_0}^{ \eta}\Ab{A}) - \rho_0^\eta (t)(A)|=\Or(\eta^\infty)  \,.
\end{equation}

The scope of the adiabatic theorem expands considerably when additive perturbations $\epsi V^\Lambda(t)$ of $H_0^\Lambda(t)$ are also taken into account, e.g., when justifying linear response theory with the help of the adiabatic theorem \cite{henheikteufel2020}.
Hence, our general results  concern the dynamics generated by 
\[
H^{\Lambda,\epsi}(t) := H_0^\Lambda(t) + \epsi V^\Lambda(t) 
\]
in the thermodynamic limit.
 Here $V^\Lambda(t)$ can be a sum-of-local-terms (SLT) operator, i.e.\ a local Hamiltonian like $H_0^\Lambda(t)$, or a Lipschitz potential, or a sum of both. 
Note that for the perturbed Hamiltonian $H^{\Lambda,\epsi}(t)$ we do not even assume a spectral gap in the bulk,  since many interesting perturbations, like   a linear potential across the sample,   necessarily close the spectral gap of 
the unperturbed system when $\Lambda$ is sufficiently large.
In this situation the ground states $\rho_0(t)$ turn into resonances $\Pi^\epsi(t)$ with life-time of order $\Or(\epsi^{-\infty})$ for the dynamics $s\mapsto \E^{\I s \mathcal{L}_{H^\epsi(t)}}$,   i.e.\ for all $n,m\in\N$ it holds that 
\[
\left|   \Pi^{ \epsi }(t) ( \E^{\I s \mathcal{L}_{H^\epsi(t)}}\Ab{A}) - \Pi^{ \epsi} (t)(A)    \right|=\Or\left(\epsi^n(1+ \epsi^m|s|^{d+1})\right)\,.
\]
These resonance states we call non-equilibrium almost-stationary states (NEASS) in this context. Our adiabatic theorem, Theorem~\ref{existenceofneass3}, then establishes the existence of super-adiabatic NEASSs $\Pi^{ \epsi,\eta}(t)$ on $\mathcal{A}$ close to $\Pi^{ \epsi }(t)$ such that 
 the adiabatic evolution $\mathfrak{U}^{\epsi, \eta}_{t,t_0}$ generated by $\frac{1}{\eta} \mathcal{L}_{H^{\epsi}(\cdot)}$   approximately intertwines the super-adiabatic NEASSs in the following sense: for any $n>d$ and  for all $A\in\mathcal{D}\subset \mathcal{A}$
\begin{equation}\label{IntroThm}
\left|\Pi^{ \epsi,\eta}(t_0)(\mathfrak{U}_{t,t_0}^{\epsi, \eta}\Ab{A}) - \Pi^{ \epsi,\eta} (t)(A)\right|=\Or\left(\eta^{n-d}+\frac{\epsi^n}{\eta^d}   \right)
\end{equation}
uniformly for $t$ in compact sets.
While for $\epsi=0$ this statement reduces to the standard adiabatic theorem \eqref{standardAT},  for $0<\epsi \ll1$ the right hand side of \eqref{IntroThm} is small if only if also $\eta$ is small, but not too small compared to $\epsi$, i.e.\ $  \epsi^{n/d}\ll \eta\ll1$ for some $n\in\N$.
Physically, this simply means that the adiabatic approximation breaks down when the adiabatic switching occurs at times that exceed the lifetime of the NEASS, an effect that has already been observed in adiabatic theory for resonances before, see, e.g., ~\cite{AF,EH}. 
 
Our results build on a number of recent developments in adiabatic theory for extended lattice systems. Bachmann et al.\ \cite{bachmann2018adiabatic} showed how locality in the form of the Lieb-Robinson bounds \cite{lieb1972finite} and the SLT-inverse of the Liouvillian \cite{hastings2005quasiadiabatic,bachmann2012automorphic} can be used to obtain uniform adiabatic approximations for extended lattice systems with spectral gaps in finite volumes. Motivated by   space-adiabatic perturbation theory \cite{PST2}, in \cite{teufel2020non} one of us   generalised the results of \cite{bachmann2018adiabatic,monaco2017adiabatic} to situations where the spectral gap is closed, e.g.\ due to perturbations by external fields, i.e.\ to NEASSs. In \cite{henheikteufel20202} we   lifted the  results of \cite{teufel2020non} to states for the infinite system by taking a thermodynamic limit, using recent advances on the thermodynamic limit of quantum lattice systems from \cite{nachtergaele2019quasi}. Finally, the   construction   of the spectral flow  in the infinite volume with a gap only in the bulk by Moon and Ogata \cite{moon2019automorphic} provided the motivation and a technical basis for the adiabatic theorem proved in this paper.

As mentioned above, the adiabatic theorem has profound implications for the validity of   linear response theory in such systems and thus also for transport in topological insulators. We will not discuss these questions here, but instead refer to the recent reviews \cite{henheikteufel2020,bachmann2020note} and to \cite{MPT}. Note however, that all previous results listed above assume a uniform spectral gap of the  Hamiltonians describing the unperturbed finite size  systems and are thus not applicable to topological insulators with boundaries. Instead periodic boundary conditions must be assumed. 
While one expects that in the absence of phase-transitions the effect of the boundaries is negligible, this must be proved and we achieve this in this paper as far as the adiabatic approximation is concerned.

From a technical point of view, the mentioned adiabatic theorems are all based on the same perturbative scheme, called adiabatic perturbation theory. The specific challenge for extended systems in the thermodynamic limit is to control that under all the algebraic operations used, the corresponding operators remain in spaces of operators with good locality properties and with a good thermodynamic limit.
The critical operation   to make adiabatic perturbation theory work as intended in \eqref{IntroThm}  is the inversion of the Liouvillian, the step at which the gap condition is needed. 
Good locality properties means, very roughly speaking, that each member $A^\Lambda$ of an operator family $\{A^\Lambda\}_{\Lambda\subset\Z^d}$ can be written as a sum $A^\Lambda = \sum_{i\in \Lambda} \Phi_i$ of uniformly bounded and uniformly local terms $\Phi_i$ and only the number of terms grows with the volume $|\Lambda|$ of $\Lambda$. We call such operators sum-of-local-term operators, or in short SLT-operators.\footnote{Operators of this kind are often called ``local Hamiltonians", but this could lead to confusion, since by far not all operators with this property, which occur in the following, play the physical role of a Hamiltonian.} 
The main technical challenge is thus the definition of suitable normed spaces of SLT-operators, within which   the adiabatic perturbation scheme  can be iterated an arbitrary number of times and, in particular, within which one can invert  the Liouvillian $\mathcal{L}_{H_0(t)}$ under the assumption of a spectral gap in the bulk. Here we rely, up to small modifications, on spaces  already used   in \cite{bachmann2018adiabatic}.  However, in order to control the thermodynamic limit   in a quantitative way that allows finally also statements about the bulk behaviour of finite systems with boundary, we   define a property called   ``having a rapid thermodynamic limit'' and need to show that this property is preserved under the relevant algebraic operations as well. 

  More explicitly, this property improves the convergence and quasi-locality estimates of the dynamics, the Liouvillian and the SLT-inverse of the Liouvillian from standard norm-estimates to norms measuring the quality of localisation (later called ``$f$-norms"), which were introduced in \cite{moon2019automorphic}. Since the spectral gap only appears in the bulk,  the Liouvillian can finally only be inverted in the thermodynamic limit (Proposition~\ref{invliou}), and thus the perturbative scheme can only yield the adiabatic theorem after taking that limit. This requires the improved quantitative control of the dynamics, the Liouvillian and the SLT-inverse of the Liouvillian in the thermodynamic limit, which is guaranteed by assuming the above property (Proposition~\ref{prop:rapidtdl}). Beside this concept, our proof of the adiabatic theorem in the bulk heavily relies on the idea of investigating the dynamics $\mathfrak{U}_{t,t_0}^{\epsi, \eta, \Lambda, \Lambda'}$ generated by $H_0^\Lambda$ and $V^{\Lambda'}$ for different system sizes $\Lambda' \subset \Lambda$ and taking $\Lambda \nearrow \Z^d$ and $\Lambda' \nearrow \Z^d$ separately. 
 
We end the introduction with a short plan of the paper. In Section~\ref{setting} we explain the mathematical setup. In Section~\ref{sec3} we state the generalised adiabatic theorem in the infinite volume as in~\eqref{IntroThm}. We also show that our general assumptions are indeed fulfilled for certain models of topological insulators. 
In Section~\ref{sec4} we state the generalised adiabatic theorem  for finite systems with a gap only  in the bulk and for observables $A$ supported sufficiently far from the boundary. The latter is the result  that applies to large but finite systems with edge states that close the spectral gap. 
Section~\ref{sec5} contains the proofs of our main results. Various technical details are collected in three appendices. 
\\[5mm]
  {\bf Acknowledgment.} J.H.~acknowledges partial financial support by the ERC Advanced Grant ``RMTBeyond" No.~101020331.

\section{The mathematical framework} \label{setting}
In this section, we briefly recall the mathematical framework relevant for formulating and proving our results. More details and explanations are provided in \cite{henheikteufel20202}. 
\subsection{The lattice and algebras of local observables}
We consider fermions with $r$ spin or other internal degrees of freedom on the lattice $\Gamma$. 
Let $\mathcal{P}_0(\Gamma) := \set{X \subset \Gamma : \vert X \vert < \infty}$ denote the set of finite subsets of $\Gamma$, where $\vert X \vert $ is the number of elements in $X$, and let 
\begin{equation*}
\Lambda_k := \set{-k, ..., +k}^d
\end{equation*}
be the centred box of size $2k$ with   metric $d^{\Lambda_k}(\cdot,\cdot)$.  This metric may differ from the standard $\ell^1$-distance $d(\cdot, \cdot )$ on $\Gamma$ restricted to $\Lambda_k$ if one considers discrete tube or torus geometries,   but satisfies the bulk-compatibility condition
\[ 
\forall k \in \mathbb{N} \ \forall x,y \in \Lambda_k: d^{\Lambda_k}(x,y) \le d(x,y) \ \  \text{and} \ \   d^{\Lambda_k}(x,y) = d(x,y) \ \text{whenever} \ d(x,y) \le k . 
\]  
For each $X \in \mathcal{P}_0(\Gamma)$ let $\mathfrak{F}_X$ be the fermionic Fock space build up from the one-body space $\ell^2(X,\mathbb{C}^r)$. The    $C^*$-algebra of bounded operators
 $\mathcal{A}_{X} := \mathcal{L}(\mathfrak{F}_X)$  is generated by the identity element $\mathbf{1}_{\mathcal{A}_X}$ and the creation and annihilation operators $a_{x,i}^*, a_{x,i}$ for $x \in X$ and $1 \le i \le r$, which satisfy the canonical anti-commutation relations. 
Whenever $X\subset X'$, then $\mathcal{A}_{X}$ is naturally embedded as a sub-algebra of $\mathcal{A}_{X'}$.

The algebra of local observables is defined as the inductive limit 
\begin{equation*}
\mathcal{A}_{\mathrm{loc}}  := \bigcup_{X \in \mathcal{P}_0(\Gamma)} 	\mathcal{A}_{X}\,,
\qquad\mbox{and its closure}\qquad
\mathcal{A} := 	\overline{\mathcal{A}_{\mathrm{loc}}}^{\Vert \cdot \Vert}
\end{equation*}
with respect to the operator norm $\Vert \cdot \Vert$ is a $C^*$-algebra, called the quasi-local algebra.
The even elements $\mathcal{A}^+ \subset \mathcal{A}$ form a $C^*$-subalgebra. 

Also, note that for any $X \in \mathcal{P}_0(\Gamma)$ the set of elements~$\mathcal{A}_{X}^N$ commuting with the number operator 
\begin{equation*}
N_{X} := \sum_{x \in X} a^*_x \cdot a_x := \sum_{x \in X} \sum_{i=1}^{r}a^*_{x,i} a_{x,i}
\end{equation*}
forms a subalgebra of the even subalgebra, i.e. $\mathcal{A}_X^N \subset \mathcal{A}_X^+ \subset \mathcal{A}_X$. 

For any  bounded, non-increasing function $f: [0,\infty) \to (0, \infty)$    with ${\lim_{t \to \infty}f(t) =0 }$ one defines  the set $\mathcal{D}_f^+$ of all $A \in \mathcal{A}^+$ such that   
\begin{equation} \label{fnorm}
\Vert A \Vert_f := \Vert A \Vert + \sup_{k \in \mathbb{N}} \left(\frac{\Vert \left(1 - \mathbb{E}_{\Lambda_k}\right)\Ab{A}\Vert }{f(k)}\right) <\infty\,.
\end{equation}
See Appendix~C of \cite{henheikteufel20202} for a    definition and summary of the relevant properties of the conditional expectation $\mathbb{E}_{\Lambda_k}$. As shown in Lemma~B.1 of \cite{moon2019automorphic}, $\mathcal{D}_f$ is a $*$-algebra and    $(\mathcal{D}_f,\|\cdot\|_f)$  a Banach space.
Although  $(\mathcal{D}_f,\|\cdot\|_f)$ is not a normed algebra, it follows from the proof of Lemma~B.1 in \cite{moon2019automorphic} that multiplication $\mathcal{D}_f\times \mathcal{D}_f\to \mathcal{D}_f$, $(A,B)\mapsto AB$, is continuous. Note that in \eqref{fnorm} and in the following we write   the argument of a   linear map  $\mathcal{A}\supset\mathcal{D}\to \mathcal{A}$ in   double brackets $\Ab{\cdot}$. This notation helps to structure complicated compositions of maps with different types of arguments.

As only  even observables will be relevant to our considerations, we will drop the superscript `$^+$' from now on and redefine $\mathcal{A} := \mathcal{A}^+$. 

\subsection{Interactions and associated operator families} \label{subsec:interactions}
An  {\em interaction on a domain $\Lambda_k$} is a map 
\begin{equation*}
\Phi^{\Lambda_k} : \left\{X \subset \Lambda_k\right\} \to  \mathcal{A}_{\Lambda_k}^N\,,\;  X \mapsto \Phi^{\Lambda_k} (X) \in \mathcal{A}_X^N
\end{equation*}
with values in the self-adjoint operators. Note that the maps $\Phi^{\Lambda_k}$ can   be extended to maps on the whole $\mathcal{P}_0(\Gamma)$ or restricted to a smaller $\Lambda_l$, trivially. 

In order to describe fermionic systems on the lattice $\Gamma$ in the thermodynamic limit one considers sequences  
$\Phi = \left(\Phi^{\Lambda_k}\right)_{k \in \mathbb{N}}$ of interactions on domains $\Lambda_k$ and calls the whole sequence an {\em interaction}.
An {\em infinite volume interaction} is a map 
\[
\Psi  : \mathcal{P}_0(\Gamma)  \to \mathcal{A}_\mathrm{loc}^N\,,\;  X \mapsto \Psi (X)\in \mathcal{A}_{X}^N \,,
\]
again with  values in the self-adjoint operators. Such an  infinite volume interaction defines a general interaction $\Psi = \big(\Psi^{ \Lambda_k}\big)_{k \in \mathbb{N}}$ by restriction, i.e.\ by setting   $\Psi^{ \Lambda_k} := \Psi|_{\mathcal{P}_0(\Lambda_k)}$. 

With any interaction $\Phi$ one associates a sequence 
  $A= (A^{\Lambda_k})_{k\in \mathbb{N} }$ of self-adjoint operators 
    \begin{equation*}
A^{\Lambda_k} := A^{\Lambda_k}(\Phi) := \sum_{X \subset \Lambda_k} \Phi^{\Lambda_k} (X) \in \mathcal{A}_{\Lambda_k}^N. 
\end{equation*}
A so called  $F$-function for the lattice $\Gamma$ can be chosen as $F_1(r) =\frac{1}{ (1+r)^{d+1}}$, which is   integrable and has a finite convolution constant. These properties remain valid for $F_{\zeta} = \zeta \cdot F_1$ if one multiplies by a bounded, non-increasing and logarithmically superadditive weight function $\zeta:[0,\infty) \to (0,\infty)$. The space of such weight functions is denoted by $\mathcal{W}$, while we write $\mathcal{S}$ for the space of weight functions decaying faster than any polynomial. An example for a function $\zeta \in \mathcal{S}$ is the exponential $\zeta(r) = \mathrm{e}^{-ar}$ for some $a >0$. 
The   finite convolution constant associated with $F_{\zeta}$ is denoted by  
\begin{equation*}
C_{\zeta} := \sup_{k \in \mathbb{N}} \sup_{x,y\in \Lambda_k} \sum_{z \in \Lambda_k} \frac{F_{\zeta}(d^{\Lambda_k}(x,z)) \ F_{\zeta}(d^{\Lambda_k}(z,y))}{F_{\zeta}(d^{\Lambda_k}(x,y))} < \infty\,. 
\end{equation*}
 For any  $\zeta \in \mathcal{W}$ and $n\in \mathbb{N}_0$, a norm on the vector space of interactions is  \begin{equation}\label{normdefinition}
	\Vert \Phi \Vert_{\zeta,n } :=  \sup_{k \in \mathbb{N}}\sup_{x,y \in \Gamma} \sum_{\substack{ X \in \mathcal{P}_0(\Gamma): \\x,y\in X}} {d^{\Lambda_k}}\mbox{-diam}(X)^n \frac{\Vert \Phi^{\Lambda_k}(X)\Vert}{F_{\zeta}(d^{\Lambda_k}(x,y))}\,.
\end{equation}
In order to quantify the difference of interactions ``in the bulk'' we also introduce for any 
interaction $\Phi^{\Lambda_l}$ on the domain $\Lambda_l$ and any $\Lambda_M\subset \Lambda_l$ the quantity 
  \begin{equation*} 
 \Vert \Phi^{\Lambda_l} \Vert_{\zeta, n,\Lambda_M} := \sup_{x,y \in \Lambda_M} \sum_{\substack{ X \subset \Lambda_M: \\x,y\in X}} \mbox{diam}(X)^n \frac{\Vert \Phi^{\Lambda_l}(X)\Vert}{F_{\zeta}(d (x,y))}\,,
\end{equation*}
where $d$ and diam now refer to the   $\ell^1$-distance on $\Gamma$.
Note that these norms depend on the sequence of metrics $d^\Lambda_k$ on the cubes $\Lambda_k$, i.e.\ on the boundary conditions. While this will in general not be made explicit in the notation, we add a superscript $^\circ$  to the norm and to the normed spaces defined below, if we want to emphasise the use of open boundary conditions, i.e.\ 
$d^\Lambda_k \equiv d$. The compatibility condition for the metrics $d^{\Lambda_k}$ implies that $\Vert \Psi \Vert_{\zeta, n } \le \Vert \Psi \Vert_{\zeta, n }^\circ $.

     Let $\mathcal{B}_{\zeta, n }$ be the Banach space of interactions with finite $\Vert \cdot \Vert_{\zeta, n }$-norm. Clearly,  $\mathcal{B}_{\zeta, n } \subset \mathcal{B}_{\zeta, m }$ whenever $n \ge m$. Not so obviously, for any pair $\zeta,\xi\in \mathcal{S}$ it holds that either $\mathcal{B}_{\zeta, n } \subset \mathcal{B}_{\xi, n }$ or $\mathcal{B}_{\xi, n } \subset \mathcal{B}_{\zeta, n }$ (or both) for all $n\in\N_0$, see Lemma~A.1 in \cite{monaco2017adiabatic}. 
If $\mathcal{B}_{\zeta, n } \subset \mathcal{B}_{\xi, n }$ for all $n\in\N_0$, then we say that $\xi$ dominates $\zeta$ and write $\zeta \prec \xi$.

An operator family $A$ is called a sum-of-local-terms (in short \SLT) operator family, if $\Phi_A \in \mathcal{B}_{1, 0 }$.
Moreover, define
\begin{align*}
\mathcal{B}_{\zeta, \infty } := \bigcap_{n \in \mathbb{N}_0} \mathcal{B}_{\zeta, n }, \quad  
\mathcal{B}_{\mathcal{S},n } := \bigcup_{\zeta \in \mathcal{S}} \mathcal{B}_{\zeta, n }, \quad \mathcal{B}_{\mathcal{S},\infty } := \bigcap_{n \in \mathbb{N}_0} \mathcal{B}_{\mathcal{S}, n }
\end{align*}
as spaces of interactions used in the sequel. The corresponding spaces of operator-families are denoted as $\mathcal{L}_{\zeta, n }$, $\mathcal{L}_{\zeta, \infty }$, $\mathcal{L}_{\mathcal{S},n }$, and $\mathcal{L}_{\mathcal{S},\infty }$. Lemma A.1 in \cite{monaco2017adiabatic} shows that the spaces $\mathcal{B}_{\mathcal{S},n }$ and therefore also $\mathcal{B}_{\mathcal{S},\infty }$ are indeed vector spaces.  

 Let $I \subset \mathbb{R}$ be an open interval. We say that a map $A: I \to \mathcal{L}_{\zeta, n }$ is smooth and bounded whenever it is term-wise and point-wise smooth in $t \in I$ and 
 if for all $i\in\N_0$ there exists $\zeta^{(i)}\in \mathcal{S}$ such that $\zeta^{(0)}=\zeta$ and 
 $\sup_{t\in I} \Vert \frac{\D^i}{\D t^i}\Phi_A(t)\Vert_{\zeta^{(i)}, n } < \infty$. The corresponding spaces of smooth and bounded time-dependent interactions and operator families are denoted by $\mathcal{B}_{I, \zeta, n }$ and $\mathcal{L}_{I, \zeta, n}$ and are equipped with the norm $\Vert\Phi \Vert_{I, \zeta, n } := \sup_{t \in I} \Vert \Phi (t)\Vert_{\zeta, n }$. 

We say that $A: I \to \mathcal{L}_{\mathcal{S}, \infty }$ is smooth and bounded, if for any $n \in \mathbb{N}_0$ there is a $\zeta_n \in \mathcal{S}$ such that $A: I \to \mathcal{L}_{\zeta_n,n }$ is smooth and bounded and we write $\mathcal{L}_{I, \mathcal{S}, \infty }$ for the corresponding space (similarly for $\mathcal{L}_{I, \zeta, \infty }$ and $\mathcal{L}_{I, \mathcal{S}, n }$).

\subsection{Lipschitz potentials}
For the  perturbation we will allow external potentials
$
v = \left( v^{\Lambda_k}: \Lambda_k \to \mathbb{R}\right)_{ k \in \mathbb{N}} 
$
that satisfy the Lipschitz condition
\begin{equation}
C_v :=  \sup_{k \in \mathbb{N}} \sup_{\substack{x,y \in \Lambda_k}} \frac{\vert v^{\Lambda_k}(x)- v^{\Lambda_k}(y)\vert}{ d^{\Lambda_k}(x,y)} < \infty, \label{C_v}
\end{equation}
and call them for short {  \it Lipschitz potentials}. With a Lipschitz potential $v$ we associate the corresponding operator-sequence $V_v = (V_v^{ \Lambda_k})_{ k \in \mathbb{N}}$ defined by 
\begin{equation*} 
V_v^{\Lambda_k} := \sum_{x \in \Lambda_k} v^{ \Lambda_k}(x) \hspace{0.5mm} a^*_x\hspace{-0.5mm} \cdot \hspace{-0.5mm}a_x. 
\end{equation*}
The space of Lipschitz potentials is denoted by $\mathcal{V}$.  As above, we add a superscript $^\circ$  to the constant and to the space of Lipschitz potentials, if we want to emphasise the use of open boundary conditions, i.e.\ 
	$d^\Lambda_k \equiv d$. The compatibility condition for the metrics $d^{\Lambda_k}$ implies that $C_v \ge C_v^\circ $. Note that a Lipschitz function $v_\infty:\Gamma\to \R$ defines, again by restriction, a Lipschitz potential in $\mathcal{V}^\circ$ and we write $v_\infty \in \mathcal{V}^\circ$ with a slight abuse of notation. We call $v_\infty$ an \textit{infinite volume Lipschitz potential}.  
	
 We say that the map $v: I \to \mathcal{V}$ is smooth and bounded whenever it is term-wise and point-wise smooth in $t \in I$ and satisfies $\sup_{t \in I} C_{ \frac{\D^i}{\D t^i}v (t)} < \infty$ for all $i \in \mathbb{N}_0$. 
 The space of smooth and bounded Lipschitz potentials on the interval $I$ is denoted by $\mathcal{V}_I$.

A crucial property of Lipschitz potentials is, that their commutator with an \SLT~operator family is again an \SLT~operator family (see Lemma 2.1 in \cite{teufel2020non}).
\begin{lem} \label{lipschitzcomm}
Let $v \in \mathcal{V}_I$ and $A \in \mathcal{A}_X$. Then we have
\begin{equation*}
\sup_{k \in \mathbb{N}} \sup_{s \in I} \left\Vert  \left[V_v^{\Lambda_k}(s), A\right] \right\Vert \le \tfrac{1}{2}C_v r \,\mathrm{diam}(X)^{d+1} \Vert A \Vert.
\end{equation*}
For $A \in \mathcal{L}_{I,\zeta,n+d+1 }$ it holds that
\begin{equation*}
\left\Vert \Phi_{[A,V_v]} \right\Vert_{I, \zeta, n } \le \tfrac{1}{2}C_v r \left\Vert \Phi_A \right\Vert_{I, \zeta, n+d+1 }
\end{equation*}
and hence, in particular,  
$
[A,V_v] \in \mathcal{L}_{I, \zeta, n }. 
$
\end{lem}

\section{The   adiabatic theorem for   infinite systems with a gap in the bulk}\label{sec3}
Subject of our investigations are time-dependent Hamiltonians of the form
\begin{equation*}
H^{\varepsilon}(t)= H_0( t) + \varepsilon V(t)
\end{equation*}
for $t\in I\subset \R$.
The unperturbed Hamiltonian $H_0(t)$, defined in terms of its short range infinite volume interaction $\Psi_{H_0}$, is assumed to have a spectral gap in the bulk uniformly in $t$, a notion to be made precise in the following. This gap  may be closed by the perturbation $V(t)=V_{v_\infty}(t) + H_1(t)$ consisting of a  Lipschitz potential and an additional short range Hamiltonian, both defined on all of $\Gamma$ through   an infinite volume Lipschitz potential $  v_\infty $  resp.\ an infinite volume interaction $\Psi_{H_1}$.

Before we can formulate the precise assumptions and the statement of the theorem of this section, we first need to discuss the definition and the properties of the infinite volume dynamics generated by $H^\epsi(t)$. 
With the infinite volume interactions one  naturally associates   sequences of finite volume systems,   $\Psi_{H_0} = \big( \Psi_{H_0}^{\Lambda_k}\big)_{k\in\N} $, $\Psi_{H_1} = \big( \Psi_{H_1}^{\Lambda_k}\big)_{k\in\N} $, $v_\infty = ( v_\infty|_{\Lambda_k})_{k\in\N}$. The corresponding finite volume Hamiltonians $H^{\epsi, \Lambda_k}(t)$ are self-adjoint and generate   unique unitary evolution families $\mathfrak{U}^{\Lambda_k}_{t,s} $ (in the Heisenberg picture) on the finite volume algebras $\mathcal{A}_{\Lambda_k}$. It is shown, e.g.,   in \cite{bratteli1996operator, bru2016lieb,nachtergaele2019quasi}   that $\mathfrak{U}^{\Lambda_k}_{t,s} $
converges strongly on $\mathcal{A}_{\rm loc}$ to a  
unique co-cycle of automorphisms $\mathfrak{U}_{t,s}  $ on $\mathcal{A}$ whenever 
  $H_0,H_1 \in \mathcal{L}_{I, \zeta, 0 }^\circ$.
However, for our proofs we need norm convergence of $\mathfrak{U}^{\Lambda_k}_{t,s}$ to $\mathfrak{U}_{t,s}  $ with respect to   $\|\cdot\|_f$-norms for suitable decay-functions $f\in \mathcal{S}$. With the next section in mind, where we consider also finite systems with a gap in the bulk, we now introduce  a new property for general interactions $\big(\Phi^{\Lambda_k}\big)_{k\in\N}$ called {\em rapid thermodynamic limit}. This property   guarantees norm convergence of the unitary evolution families generated by the corresponding operators.
 It extends the  standard notion of having a thermodynamic limit  as used, e.g.,  in Definition~2.1 in \cite{henheikteufel20202},
by a quantitative control on the rate of convergence.
Definition~\ref{cauchydefinition} might seem unnecessarily complicated at first sight, but it balances two requirements: 
On the one hand, the condition  is strong enough, such that for interactions with a rapid thermodynamic limit one can show that the induced finite volume dynamics,   derivations and the SLT inverse  of the Liouvillian converge to their infinite volume limits in the norm of bounded operators 
 between  normed spaces $\mathcal{D}_{f_1}, \mathcal{D}_{f_2} \subset\mathcal{A}$ of quasi-local observables with appropriate  $f_1,f_2 \in \mathcal{S}$, cf.\
Proposition~\ref{prop:rapidtdl} below. Moreover, one   obtains sufficient control of the rate of convergence of these induced operations within the bulk (cf.\ Appendix~\ref{technicallemmata}) in order to conclude also an adiabatic theorem in  the bulk of a finite volume system as formulated in  Theorem~\ref{existenceofneass4}.
On the other hand, the condition of having a rapid thermodynamic limit is sufficiently weak, so that it is preserved under all the operations used in the perturbative  construction of super-adiabatic states, cf.\ Appendix~\ref{app:C}.
Finally, it is satisfied by all physically meaningful models known to us and, most importantly for the present section, it is satisfied by interactions arising from restricting an infinite volume interaction.

\begin{defi}{\rm (Rapid thermodynamic limit of interactions and potentials) }\label{cauchydefinition}
	\begin{itemize}
		\item[(a)] A time-dependent interaction $\Phi  \in \mathcal{B}_{I,\zeta, n }$ is said to \textit{have a rapid thermodynamic limit with exponent $\gamma \in (0,1)$}, if 
	 there exists an infinite volume interaction $\Psi  \in \mathcal{B}_{I,\zeta, n }^\circ$
		such that 
	\begin{align}
		\forall i \in \mathbb{N}_0 \ \ &\exists \lambda,  C>0 \ \ \forall M\in \N\ \ \forall \,   k \, \ge \, M + \lambda M^{\gamma} : \label{rapidtdl}\\ &\sup_{t \in I}\left\Vert  \frac{\D^i}{\D t^i}\left( \Psi- \Phi^{ \Lambda_k}\right) (t) \, \right\Vert_{\zeta^{(i)}, n,\Lambda_M}  \le C \, \zeta^{(i)}(M^{\gamma}) =: C \zeta^{(i)}_{\gamma}(M). \nonumber
		\end{align}
		We write $\Phi\stackrel{\rm r.t.d.}{\rightarrow} \Psi$ in this case.
		Note, that $\zeta \in \mathcal{S}$ implies $\zeta_{\gamma} \in \mathcal{S}$.

		A time-dependent interaction $\Phi  \in \mathcal{B}_{I,\mathcal{S}, n }$ is said to \textit{have a  rapid thermodynamic} limit with exponent $\gamma \in (0,1)$ if there exists   $\zeta_n \in \mathcal{S}$ such that $\Phi  \in \mathcal{B}_{I,\zeta_n, n }$ has a rapid thermodynamic limit with exponent $\gamma \in (0,1)$; a time-dependent interaction $\Phi  \in \mathcal{B}_{I,\mathcal{S}, \infty }$ is said to \textit{have a rapid thermodynamic limit} with exponent $\gamma \in (0,1)$ if for any $n \in \mathbb{N}$ exists $\zeta_n \in \mathcal{S}$ such that $\Phi  \in \mathcal{B}_{I,\zeta_n, n }$ has a rapid thermodynamic limit with exponent $\gamma \in (0,1)$ (similarly for $\mathcal{B}_{I, \zeta, \infty }$ and $\mathcal{B}_{I, \mathcal{S}, n }$). 
		
		A family of operators is said to \textit{have a rapid thermodynamic limit} if and only if the corresponding interaction does. 
		\item[(b)] 	A Lipschitz potential 
		$v \in \mathcal{V}_I$
		is said to \textit{have a rapid thermodynamic limit with exponent $\gamma \in (0,1)$} if it is eventually independent of $k$, i.e.~if  there exists  an infinite volume Lipschitz potential $ v_\infty \in \mathcal{V}_I^\circ$ 		such that 
		\begin{align*}
		\exists \lambda >0\ \ \forall M\in \N \ \ \forall \,  k \, \ge \, M + \lambda M^{\gamma}, t \in I: \ \  v_{ \infty}(t, \cdot) |_{\Lambda_M} = v^{ \Lambda_k}(t, \cdot) |_{\Lambda_M} \,. 
		\end{align*}
		We write $v\stackrel{\rm r.t.d.}{\rightarrow} v_\infty$ in this case. 
			\end{itemize}
\end{defi}
Trivially, an infinite volume interaction $\Psi  \in \mathcal{B}_{I,\zeta, n }^\circ$ has a rapid thermodynamic limit with any exponent $\gamma\in(0,1)$.
Moreover,  an interaction resp.\ a Lipschitz potential having a rapid thermodynamic limit with exponent $\gamma \in (0,1)$ clearly has a thermodynamic limit in the sense of Definition~2.1 from \cite{henheikteufel20202} (see also Definition 3.7 from \cite{nachtergaele2019quasi}), and   thus the property of having a  rapid  thermodynamic limit for the interaction resp.\ potential guarantees also the existence of a thermodynamic limit for the associated evolution families   and derivations as in Proposition~2.2 from \cite{henheikteufel20202} (see also Proposition~\ref{lrb} and Proposition~\ref{tdlofderivations}).

 Definition~\ref{cauchydefinition} also applies to time-independent interactions and potentials trivially, and  the interactions in $\mathcal{B}_{I,\zeta, n }$ and that the Lipschitz potentials in $\mathcal{V}_I$ that have a rapid thermodynamic limit with exponent $\gamma \in (0,1)$ form sub-vectorspaces.  We provide an equivalent characterisation of interactions having a rapid thermodynamic limit in terms of a Cauchy condition in Lemma \ref{lem:equivalentchar}.

  The improved quasi-locality estimates as required for the proof of the adiabatic theorem in the infinite volume, Theorem~\ref{existenceofneass3}, are summarised in the following proposition. More detailed  statements, which make explicit   the rate of convergence  
 and  also apply   to the SLT-operator families  that appear in the perturbative construction of the NEASS, 
  are given in Appendix~\ref{technicallemmata}. They are, however, only required for the proof of the adiabatic theorem in finite volume, Theorem~\ref{existenceofneass4}.
 \begin{prop}{\rm (Rapid thermodynamic limit and quasi-locality estimates)} \label{prop:rapidtdl} \\
	Let $H_0 \in \mathcal{L}_{I, \exp(-a \,\cdot), 0 }$   and $v \in \mathcal{V}_I$ both have a rapid thermodynamic limit with exponent $\gamma \in (0,1)$, set $H=H_0 + V_v$ and let $\mathfrak{U}_{t,s}$ denote the bulk dynamics generated by the derivation (the Liouvillian) $\mathcal{L}_{H (t)}: D(\mathcal{L}_{H (t)}) \to \mathcal{A}$ associated with $H(t)$ at $t \in I$. Moreover, let $\mathcal{I}_{H_0 (t)} $
	be the SLT-inverse  of the Liouvillian $\mathcal{L}_{H_0(t)}$ (see Appendix~\ref{invliouappendix}).  The finite volume versions for the cubes $\Lambda_k$ are indicated by the superscript $\Lambda_k$. 
	
	Then the dynamics, the Liouvillian, and the SLT-inverse   of the Liouvillian preserve quasi-locality and the finite volume versions converge rapidly to their bulk versions in the following sense:  For suitable functions $f_1, f_2 \in \mathcal{S}$ (which may be different in every case below)   it holds that:
	\begin{enumerate}
	\item 	$\mathfrak{U}_{t,s} :\mathcal{D}_{f_1} \to  \mathcal{D}_{f_2}$ is a bounded operator and $\mathfrak{U}^{\Lambda_k}_{t,s} \xrightarrow{k \to \infty}\mathfrak{U}_{t,s}$ in the corresponding operator norm, both uniformly for $s,t$ in compacts. 
	\item 	$\mathcal{D}_{f_1} \subset D(\mathcal{L}_{H (t)})$, $\mathcal{L}_{H (t)} :\mathcal{D}_{f_1} \to  \mathcal{D}_{f_2}$ is a bounded operator and $\mathcal{L}_{H (t)}^{\Lambda_k} \xrightarrow{k \to \infty}\mathcal{L}_{H_ (t)}$ in the corresponding operator norm, both uniformly for $t \in I$. 
		\item 	$\mathcal{I}_{H_0 (t)} :\mathcal{D}_{f_1} \to  \mathcal{D}_{f_2}$ is a bounded operator and $\mathcal{I}_{H_0 (t)}^{\Lambda_k} \xrightarrow{k \to \infty} \mathcal{I}_{H_0 (t)}$ in the corresponding operator norm, both uniformly for $t \in I$. 
	\end{enumerate}
\end{prop} 
 The gap assumption for $H_0$ in the bulk can be formulated via the Gelfand-Naimark-Segal (GNS) construction. 
 Let $\Psi_{H_0}\in \mathcal{B}_{\zeta,0 }^\circ$ be an infinite volume interaction  and let $\mathcal{L}_{H_0}$ denote the induced derivation on $\mathcal{A}$. A state $\omega$ on $\mathcal{A}$ is called a $ \mathcal{L}_{H_0}$-ground state, iff   $\omega(A^*\mathcal{L}_{H_0}(A)) \ge 0$ for all $A \in D(\mathcal{L}_{H_0})$. 
Let $\omega$ be a $  \mathcal{L}_{H_0}$-ground state and $(\mathcal{H}_{\omega}, \pi_{\omega}, \Omega_{\omega})$ be the corresponding GNS triple ($\mathcal{H}_{\omega}$ a Hilbert space, $\pi_{\omega}: \mathcal{A} \to \mathcal{B}(\mathcal{H}_{\omega})$ a representation, and $\Omega_{\omega} \in \mathcal{H}_{\omega}$ a cyclic vector). Then there exists a unique densely defined, self-adjoint positive operator $H_{0,\omega} \ge 0$ on $\mathcal{H}_{\omega}$ satisfying
\begin{equation}
\pi_{\omega}(\mathrm{e}^{\mathrm{i}t\mathcal{L}_{H_0}}\Ab{A})= \mathrm{e}^{\mathrm{i}tH_{0,\omega}} \pi_{\omega}(A)\mathrm{e}^{- \mathrm{i}tH_{0,\omega}}\quad \text{and} \quad \mathrm{e}^{- \mathrm{i}tH_{0,\omega}} \Omega_{\omega} = \Omega_{\omega}
\end{equation}
for all $A \in \mathcal{A}$ and $t \in \mathbb{R}$. We call this $H_{0,\omega}$ the bulk Hamiltonian associated with $\Psi_{H_0}$ and $\omega$. See~\cite{bratteli1996operator} for the general theory.

 We can now formulate the assumptions of the main theorem. Recall that $I \subseteq \mathbb{R}$ denotes an open time interval. \\[2mm]
 \noindent {\bf Assumptions on the interactions.} \\[-6mm]
\begin{assum} 
\begin{itemize}
	\item[(I1)] Let $\Psi_{H_0}, \Psi_{H_1} \in  \mathcal{B}_{I, \exp(-a \, \cdot),\infty }^\circ$ be time-dependent  infinite volume interactions and  $v_\infty\in\mathcal{V}^\circ_I$ a time-dependent infinite volume Lipschitz potential.  
\item[(I2)] Assume that the map $I\to  \mathcal{B}_{\exp(-a \, \cdot),\infty }^\circ$, $t\mapsto \Psi_{H_0(t)}$, is continuously differentiable.\footnote{Note that this does not follow from $\Psi_{H_0} \in \mathcal{B}_{I, \exp(-a \, \cdot),\infty }^\circ$, as the spaces of smooth and bounded interactions are defined via term-wise and point-wise time derivatives (cf.~Section \ref{subsec:interactions}).}  
\end{itemize}
\end{assum}
Moreover, we assume that $\Psi_{H_0}$ has a unique gapped ground state in the following sense: \\[2mm]
\noindent {\bf Assumptions on the ground state of $\Psi_{H_0}$.} \\[-6mm]
\begin{assum} 
\begin{itemize}
	\item[(G1)] {\bf Uniqueness}. For each $t \in I$,  there exists a unique $\mathcal{L}_{H_0(t)}$-ground state~$\rho_0(t)$.
	\item[(G2)] {\bf Gap}. There exists $g>0$, such that $\sigma(H_{0,\rho_0(t)}(t)) \setminus \{0\} \subset [2g,\infty)$ for all~$t \in I$.
	\item[(G3)] {\bf Regularity}. For any $f \in \mathcal{S}$ and $A \in \mathcal{D}_f$, $t\mapsto \rho_0(t)(A)$ is differentiable and there exists a constant $C_f$ such that for all $A \in \mathcal{D}_f$
	\begin{equation*}
	\sup_{t \in I} \vert \dot{\rho_0}(t)(A)\vert \le C_f \Vert A\Vert_f\,.
	\end{equation*}
\end{itemize}
\end{assum} 

The assumptions  (I1) and (I2) on the interactions are quite explicit and easily checked for concrete models. Moreover, they are very general and hold for standard tight-binding models with short range interactions. 
The assumptions (G1)--(G3) on the ground state of $\Psi_{H_0}$ are much harder to verify in concrete models.
Assuming uniqueness, i.e.\   (G1), the gap assumption   (G2) holds, in particular,  if the finite volume Hamiltonians have a uniform gap:
that the gap persists in the thermodynamic limit was shown, e.g., in Proposition 5.4 in~\cite{bachmann2016lieb}.
 And also the smoothness of expectation values of (almost) exponentially localised observables as under item (G3) is a consequence of a uniform gap in finite volumes (see Remark~4.15 in~\cite{moon2019automorphic} {and Lemma~6.0.1 in \cite{moondiss}}).  The uniqueness of the ground state as required under  (G1) has been shown, to our present knowledge, only in very specific quantum spin systems.  We will discuss the validity of our assumptions in more detail below.

The key role of the spectral gap condition is that it allows to construct an SLT-inverse of the Liouvillian $\mathcal{L}_{H_0(t)}$.
Assuming a gap only in the bulk, as we do, means that 
 the action of the Liouvillian can only be \textit{inverted in the bulk}. 
The following proposition shows that this is indeed a meaningful concept. Its proof is given in Appendix~\ref{invliouappendix}. 
\begin{prop}{\rm (Inverting the Liouvillian in the bulk) }\label{invliou}~\\ 
	Let $\Psi_{H_0}  \in  \mathcal{B}_{ \exp(-a \cdot), 0 }^\circ$ be an infinite volume interaction and $\rho_0$ a $\mathcal{L}_{H_0}$-ground state that satisfies the gap assumption (G2).
	Then for all $A \in \mathcal{A}$ with $\mathcal{I}_{H_0}(A) \in D(\mathcal{L}_{H_0})$ and all $B \in  \mathcal{A}$  we have 
	\begin{equation*} 
	 \rho_0\big(\big[\mathcal{L}_{H_0}\circ \mathcal{I}_{H_0}\Ab{A}- \I A,B\big]\big) = 0\,.
	\end{equation*} 
\end{prop}
 We now state our main theorem. Its proof is given in Section~\ref{sec5} and based on the statements summarised in Proposition~\ref{prop:rapidtdl} and Proposition~\ref{invliou}. 
\begin{thm}{\rm (Adiabatic theorem for infinite systems with a gap in the bulk)} \label{existenceofneass3}~\\
 Let $\Psi_{H_0}$,  $\Psi_{H_1}$ and $v_\infty$ satisfy  (I1), (I2),  (G1), (G2), and (G3). 
	 Denote by  $\mathfrak{U}_{t,t_0}^{\varepsilon, \eta}$   the Heisenberg time-evolution on $\mathcal{A}$ generated by 
	 $\frac{1}{\eta}\Psi_{H^{\varepsilon}}$ with 
	 \[
	 \Psi_{H^{\varepsilon}}:= \Psi_{H_0} + \epsi (V_{v_\infty} + \Psi_{H_1})
	 \] and 
	 adiabatic parameter $\eta \in (0,1]$, cf.\  Corollary~\ref{lrbcorr}.
	Then for any $\varepsilon, \eta \in (0,1]$ and $t \in I$ there exists a near-identity automorphism $\beta^{\varepsilon, \eta}(t)= \mathrm{e}^{\mathrm{i} \varepsilon\mathcal{L}_{S^{\varepsilon, \eta}(t)}}$ with $\Psi_{  S^{\varepsilon, \eta}}\in \mathcal{B}^\circ_{I,\mathcal{S},\infty} $
		such that the super-adiabatic NEASS defined by 
	\begin{equation*}
	\Pi^{\varepsilon, \eta}(t) := \rho_0(t)\circ \beta^{\varepsilon, \eta}(t)  
	\end{equation*}
	has the following properties:
	\begin{enumerate}
\item $\Pi^{\varepsilon, \eta}$ almost {\bf intertwines the   time evolution}: For any {$n \in \mathbb{N}$}, there exists a constant $C_n$ such that for any $ t \in I$ and for all finite $X \subset \Gamma$ and $A \in \mathcal{A}_X \subset \mathcal{A}$
		\begin{align} \label{adiabboundbulk}
		&\left\vert \Pi^{\varepsilon, \eta}(t_0)(\mathfrak{U}_{t,t_0}^{\varepsilon, \eta}\Ab{A}) - \Pi^{\varepsilon, \eta}(t)(A) \right\vert \nonumber \\ &\qquad \qquad  \le \ C_n \ \frac{\varepsilon^{n+1} + \eta^{n+1}}{\eta^{d+1}}   \ \left(1+\vert t- t_0\vert^{d+1}\right) \,  \Vert A \Vert \, \vert X \vert^2.
		\end{align}
	For any $f \in\mathcal{S}$  the same statement (with a different constant $C_n$) holds also for all  $A \in \mathcal{D}_f$ after replacing $\Vert A \Vert \vert X \vert^2$ by $ \Vert A \Vert_f$. 
\item $\Pi^{\varepsilon, \eta}$  is  {\bf local in time}:
		$\beta^{\varepsilon, \eta}(t)$ depends only on $\Psi_{H^{\varepsilon}}$ and its time derivatives at time~$t$. 
\item $\Pi^{\varepsilon, \eta}$  is \textbf{stationary} whenever the Hamiltonian is stationary: 
  if for some $t\in I$ all time-derivatives of $\Psi_{H^{\varepsilon} }$ vanish at time $t$,
then $\Pi^{\varepsilon, \eta}({t}) = \Pi ^{\varepsilon, 0}({t})$. 
\item $\Pi^{\varepsilon, \eta} $  equals the ground state $\rho_0$  of $\mathcal{L}_{H_0 }$ whenever the perturbation vanishes and the Hamiltonian is stationary:
  if for some $t\in I$ all time-derivatives of $\Psi_{H^{\varepsilon} }$   vanish at time $t$ and $\Psi_{H_1(t)} = 0$ and $v_\infty(t) =0$, 
		then $\Pi^{\varepsilon, \eta}(t) = \Pi^{\varepsilon, 0}(t) = \rho_0(t)$.  
	\end{enumerate}
\end{thm}

Let us discuss a class of Hamiltonians that are used to model Chern  or topological insulators and for which  our assumptions  can be shown  or are expected  to hold. Consider the unperturbed (time independent) Hamiltonian 
\begin{align} \label{examplehamiltonian2}
H_0^{\Lambda_k} \;= \;&\sum_{x,y \in \Lambda_k} a_x^*T(x -y)a_y \;+\; \sum_{x \in \Lambda_k} a_x^*\phi(x)a_x \nonumber\\ &+ \sum_{x,y \in \Lambda_k} a_x^* a_x W(d^{\Lambda_k}(x,y))a_y^*a_y\; -\; \mu N_{\Lambda_k}.
\end{align}
Assume that the kinetic term $T : \Gamma \to \mathcal{L}(\mathbb{C}^r)$ is an exponentially  decaying function with $T(-x) = T(x)^*$, the potential term $\phi: \Gamma \to \mathcal{L}(\mathbb{C}^r)$ is a bounded function taking values in the self-adjoint matrices, and the two-body interaction $W: [0,\infty) \to \mathcal{L}(\mathbb{C}^r)$ is exponentially decaying and also takes values in the self-adjoint matrices. Note, that $x-y$ in the kinetic term refers to the difference modulo $\Lambda_k$ if $\Lambda_k$ is supposed to have a torus geometry. 

It is well known  that   non-interacting Hamiltonians $H_0$, i.e.\ with $W \equiv 0$, of the type \eqref{examplehamiltonian2}   on a {\it torus} (periodic boundary condition)    have a spectral gap whenever the chemical potential $\mu$ lies in a gap of the spectrum of the corresponding one-body operator on the infinite domain.  It was recently shown \cite{hastings2019stability,DS}, that  the spectral gap remains open when perturbing by  sufficiently small short-range interactions $W \neq 0$. On the other hand,  the Hamiltonian $H_0$ on a {\it cube}  with open boundary condition has, in general, no longer a spectral gap because of the appearance of edge states.

Note that  for either boundary condition the corresponding interactions have a rapid thermodynamic limit and converge to the same infinite volume interaction $\Psi_{H_0}$ defined by $T$, $\phi$ and $W$ on all of $\Gamma$. Thus, 
in the thermodynamic limit the different boundary conditions lead to the same infinite volume dynamics on the quasi-local algebra $\mathcal{A}$. 
More precisely, by application of, e.g., Proposition~2.2 from \cite{henheikteufel20202}, both sequences of dynamics are strongly convergent, i.e. 
\begin{equation*}
\E^{\I t \mathcal{L}^{\Lambda_k}_{H_{0,\mathrm{torus}}}}\Ab{A} \xrightarrow{k \to \infty} \E^{\I t \mathcal{L}_{H_{0 }}}\Ab{A}
 \xleftarrow{k \to \infty} \E^{\I t \mathcal{L}^{\Lambda_k}_{H_{0,\mathrm{cube}}}}\Ab{A}
\end{equation*}
for any observable $A\in \mathcal{A}_{\mathrm{loc}}$.
 
The spectral gap for the finite volume ground states (which we have for the torus geo\-metry) implies a spectral gap also for the GNS Hamiltonian associated to any   limiting 
$  \mathcal{L}_{H_0}$-ground state,
 as it cannot close abruptly in the thermodynamic limit (see Proposition~5.4 in \cite{bachmann2016lieb}). Therefore, the Hamiltonian \eqref{examplehamiltonian2} with open boundary conditions has a rapid thermodynamic limit and a gap \textit{only in the bulk}. While also uniqueness of the ground state
 is expected to hold for such models, to our knowledge it has been shown only for certain types of spin systems, cf.\ \cite{yarotsky2005uniqueness, FP,nachtergaele2020quasi, nachtergaele2021stability,henheikteufelwessel2021}.  
  
 A typical perturbation that is physically relevant  is the adiabatic switching of a small constant electric field.
 Theorem~\ref{existenceofneass4} then shows that the response in the bulk is local and independent of the boundary condition. For more details on the implications for linear response we refer to \cite{henheikteufel2020}.
  
\section{The adiabatic theorem for finite domains}\label{sec4}
Let $\Phi_{H_0}, \Phi_{H_1} \in  \mathcal{B}_{I, \exp(-a \, \cdot),\infty } $ be  interactions  and $v\in \mathcal{V}_I$ a Lipschitz potential, all  with a rapid thermodynamic limit, i.e.\ $\Phi_{H_0} \stackrel{\rm r.t.d.}{\rightarrow} \Psi_{H_0}$, $\Phi_{H_1} \stackrel{\rm r.t.d.}{\rightarrow} \Psi_{H_1}$, and $v \stackrel{\rm r.t.d.}{\rightarrow} v_\infty$.
For the infinite volume interaction $\Psi_{H_0}$ we require again existence of a unique gapped ground state as made precise in assumptions (G1)--(G3). However, in this section we prove an adiabatic theorem for the finite volume dynamics of states $\rho_0^{\Lambda_k}(t)$ that are close to the infinite volume ground state $\rho_0(t)$ in the bulk. 
More precisely, we require
\begin{assum}
 \begin{itemize}
	\item[(S)] The sequence $\big(\rho_0^{\Lambda_k}(t)\big)_{k\in \N}$   of  states  on $\mathcal{A}_{\Lambda_k}$ converges rapidly to $\rho_0(t)$ in the bulk:     there exist $C\in\R$, $m  \in \mathbb{N}$ and $h \in \mathcal{S}$ such that   for any finite $X \subset \Gamma$, $A \in \mathcal{A}_X$, and $\Lambda_k\supset X$
\begin{equation*}
\sup_{t \in I} \left\vert \rho_0(t)(A)  - \rho_0^{\Lambda_k}(t)(A) \right\vert \le C \Vert A \Vert \,\mathrm{diam}(X)^m \,  h\big(\mathrm{dist}(X,\Gamma \setminus \Lambda_k)\big)\,.
\end{equation*}
\end{itemize}
\end{assum}

While   the sequence $\rho_0(t)|_{\mathcal{A}_{\Lambda_k}}$ obtained by just restricting the infinite volume ground state trivially satisfies conditions (S), the adiabatic theorem below holds for super-adiabatic NEASS constructed from any sequence satisfying (S). Most interesting for application would be  a sequence of ground states 
$\rho_0^{\Lambda_k}(t)$   of the finite volume Hamiltonians $H_0^{\Lambda_k}(t)$.
While (S)  is expected to hold for any sequence of finite volume ground states for  the models discussed in the previous section, the only result we are aware of indeed proving such a statement is again for weakly interacting spin systems  \cite{yarotsky2005uniqueness}.

The following theorem asserts    that by   assuming a spectral gap in the bulk for the infinite system 
one obtains similar   adiabatic bounds also for states of  finite systems (without a spectral gap!) that are close to the infinite volume ground state in the bulk. As adiabaticity need not hold at the boundaries of the finite systems, non-adiabatic effects arising at the boundary may propagate into the bulk.  For this reason, an additional error term appears, but  it decays faster than any polynomial in the size of the finite system for any  fixed $\eta$.

\begin{thm}{\rm (Adiabatic theorem for finite systems with gap in the bulk)}\label{existenceofneass4}~\\
Let $\Phi_{H_0}, \Phi_{H_1} \in  \mathcal{B}_{I, \exp(-a \, \cdot),\infty } $ be  interactions  and $v\in \mathcal{V}_I$ a Lipschitz potential, all  with a rapid thermodynamic limit, i.e.\ $\Phi_{H_0} \stackrel{\rm r.t.d.}{\rightarrow} \Psi_{H_0}$, $\Phi_{H_1} \stackrel{\rm r.t.d.}{\rightarrow} \Psi_{H_1}$, and $v \stackrel{\rm r.t.d.}{\rightarrow} v_\infty$.
Let $\Psi_{H_0}$, $\Psi_{H_1}$, and $ v_\infty$ satisfy (I1), (I2), (G1), (G2), and (G3). Moreover, let $\big(\rho_0^{\Lambda_k}(t)\big)_{k\in \N}$ satisfy~(S).
For all $k\in\N$ let $\mathfrak{U}_{t,t_0}^{\varepsilon, \eta, \Lambda_k}$ be the Heisenberg time-evolution on $\mathcal{A}_{\Lambda_k}$ generated~by 
\[
\tfrac{1}{\eta}H^{\varepsilon,\Lambda_k}(t) := \tfrac{1}{\eta} \left( H^{ \Lambda_k}_0(t) + \epsi \left( 
V_v^{\Lambda_k} (t) + H^{\Lambda_k}_1(t) \right)\right)
\]
 with adiabatic parameter $\eta \in (0,1]$. 
	Then for any $\varepsilon, \eta \in (0,1]$ and $t \in I$ there exists a near-identity automorphism $\beta^{\varepsilon, \eta,\Lambda_k}(t)$ of $\mathcal{A}_{\Lambda_k}$
		such that the super-adiabatic NEASS defined by 
	\begin{equation*}
	\Pi^{\varepsilon, \eta,\Lambda_k}(t)  := \rho_0^{\Lambda_k}(t)\circ  \beta^{\varepsilon, \eta,\Lambda_k}(t)   
	\end{equation*}
	has the following properties:
	\begin{enumerate}
		\item It almost {\bf intertwines the   time evolution of observables in the bulk}: There exists $\lambda >0$ such that for any $n \in \mathbb{N}$  there exists a constant  $C_n$ {and for any compact $K\subset I$ and $m\in \N$ there exists a constant  $\tilde{C}_{n,m,K}$} such that for all {$k\in \N$}, all  finite $X \subset \Lambda_k$,   all $A \in \mathcal{A}_X$, and all $t,t_0\in K$  
		\begin{align} \label{adiabboundfinite}
		&\left\vert \Pi^{\varepsilon, \eta, \Lambda_k}(t_0)(\mathfrak{U}_{t,t_0}^{\varepsilon, \eta, \Lambda_k}\Ab{A}) - \Pi^{\varepsilon, \eta, \Lambda_k}(t)(A) \right\vert  \\ \nonumber
		& \qquad \qquad \qquad \le \ C_n \, \frac{\varepsilon^{n+1} + \eta^{n+1}}{\eta^{d+1}} \,  \left(1+\vert t- t_0\vert^{d+1}\right) \  \Vert A \Vert \, \vert X \vert^2 \\ \nonumber
		& \qquad \qquad \qquad \quad \ + \ \tilde{C}_{ n,m,K} \,  \left(1+\eta \,  \mathrm{dist}(X,\Gamma \setminus \Lambda_{\lfloor k -\lambda k^{\gamma}\rfloor} )\right)^{-m}\Vert A \Vert \,\mathrm{diam}(X)^{2d}. 
		\end{align}
		\item It is  {\bf local in time}:
		$\beta^{\varepsilon, \eta,\Lambda_k}(t)$ depends only on $H^{\varepsilon,\Lambda_k}$ and its time derivatives at time $t $. 
		\item It is \textbf{stationary} whenever the Hamiltonian is stationary: 
		  if for some $t\in I$ all time-derivatives of $H^{\varepsilon,\Lambda_k} $ vanish at time $t$,
			then $\Pi^{\varepsilon, \eta,\Lambda_k}({t}) = \Pi ^{\varepsilon, 0,\Lambda_k}({t})$. 
		\item It equals $ \rho^{\Lambda_k}_0(t)$ whenever the perturbation vanishes and the Hamiltonian is stationary:
		  if for some $t\in I$ all time-derivatives of $H^{\varepsilon, \Lambda_k} $   vanish at time $t$ and $V^{\Lambda_k}(t)=0$, 
			then $\Pi^{\varepsilon, \eta,\Lambda_k}(t) = \Pi^{\varepsilon, 0,\Lambda_k}(t) = \rho^{\Lambda_k}_0(t)$.  
	\end{enumerate}
\end{thm}

As mentioned before, the second term in the bound \eqref{adiabboundfinite} arises when non-adiabatic transitions that can occur in the boundary region propagate into the bulk. While we do not prove this explicitly, our proof suggests that if the perturbations $v$, $H_1$ and $\dot H_0$ are supported in the bulk and if $\rho_0^{\Lambda_k}(t)$ is the ground state of $H_0(t)$, then this second part of the bound can be replaced by a term that is uniform in time and $\eta$. Finally, for the topological insulators discussed in the previous section one expects that excitations at the boundary cannot propagate into the bulk, since the chiral edge states do not couple to states supported in the bulk. I.e., in such systems   we expect \eqref{adiabboundfinite} to hold with the second term replaced by a bound that  is uniform in time and $\eta$ even if the adiabatic perturbation acts also near the boundary.

\section{Proofs of the main results}\label{sec5}
\subsection{Proof of the adiabatic theorem in the bulk}
\begin{proof} [Proof of Theorem~\ref{existenceofneass3}]
 The proof  is an adaption of the strategies  for proving the analogous theorem  for finite domains (see Theorem 5.1 in \cite{teufel2020non}) and its thermodynamic limit version (see Theorem 3.2 in \cite{henheikteufel20202}). In order to cope with the new situation of  ``a gap only in the bulk'',  we use and adapt techniques originally developed in \cite{moon2019automorphic} and summarised in a modified way in Appendix~\ref{technicallemmata}.
	
Before going into details, let us briefly outline the structure of the proof. We  first  construct a sequence of densely defined derivations  $(\mathcal{L}_{A_j^{\varepsilon, \eta}(t)})_{j\in\N}$  such that for each $n\in\N$  the state $\Pi_n^{\epsi,\eta}(t) := \rho_0(t)\circ \beta_n^{\varepsilon, \eta}(t)$ constructed from
	 the automorphism 
	  \begin{equation*}
\beta_n^{\varepsilon, \eta}(t)\Ab{A} := \mathrm{e}^{-\mathrm{i}\varepsilon \mathcal{L}_{S_n^{\varepsilon, \eta}(t)}}\Ab{A}  
\end{equation*}
generated by
	 \begin{equation}\label{SnDef}
\varepsilon \mathcal{L}_{S_n^{\varepsilon, \eta}(t)} := \sum_{j = 1}^{n} \varepsilon^{j} \mathcal{L}_{A_j^{\varepsilon, \eta}(t)}
\end{equation}
		 satisfies \eqref{adiabboundbulk}.
	The $n$-independent states $\Pi^{\epsi,\eta}(t)$ are then obtained from a suitable resummation of the asymptotic series \eqref{SnDef} for $n\to\infty$. 
	
In 	\cite{teufel2020non} the corresponding SLT-operators $A_j^{\varepsilon, \eta, \Lambda_l}(t)$ on finite domains $\Lambda_l$ are constructed from an algebraic iterative scheme. Basically, we would like  to apply the same scheme to construct the associated derivations  in the infinite volume with good control of the remainders that add up to the right hand side of \eqref{adiabboundbulk}. To do so, we first   construct the coefficients $A_j^{\varepsilon, \eta, \Lambda_l,\Lambda_k}(t)$ in finite volumes $\Lambda_l$ exactly as in \cite{teufel2020non}, where the perturbation is restricted to an even smaller volume $\Lambda_k$. Because of the absence of a spectral gap in finite volume, however,  the remainders are not   small, i.e.\ a bound like \eqref{adiabboundbulk} does not hold.  
The two main steps are now
\begin{itemize}
\item[(1)] to show that the SLT-operators   $A_j^{\varepsilon, \eta, \Lambda_l, \Lambda_k}(t)$ have a rapid thermodynamic limit.\footnote{For this proof of the adiabatic theorem in the infinite volume it would suffice to show that the  $A_j^{\varepsilon, \eta, \Lambda_l, \Lambda_k}(t)$ have a   thermodynamic limit. The more precise control on the rate of convergence will be used only in the proof of the adiabatic theorem in finite volume, Theorem~\ref{existenceofneass4}.}
\item[(2)] to show that \eqref{adiabboundbulk} holds in the thermodynamic limit.
\end{itemize}
Part (1) is rather technical: 
The iterative construction of the operators $A_j^{\varepsilon, \eta, \Lambda_l,\Lambda_k}(t)$ takes as an initial input the instantaneous Hamiltonian $H^{  \Lambda_l}_0(t)$,  its time derivative $\dot H^{  \Lambda_k}_0(t)$, and the perturbation $V^{\Lambda_k}(t)$. Then it proceeds (see   \eqref{A1} and \eqref{Aj} below) by  
 taking sums of commutators, further time-derivatives,  and by applying the map $\mathcal{I}$  defined in Appendix~\ref{inverseL}   and called, for simplicity, the inverse Liouvillian, although it does not have this property in finite volumes.  
  In Appendix~\ref{app:C} we show that all these operations preserve the property of   having a rapid thermodynamic limit (see Lemma~\ref{cauchy1}, Lemma~\ref{cauchy2} and Lemma~\ref{cauchy3}), and thus also all  $A_j^{\varepsilon, \eta, \Lambda_l,\Lambda_k}(t)$ have a rapid thermodynamic limit. 

Part (2) is now discussed in detail and we divide it into four steps.\\[3mm] 
\noindent  \textbf{(a)} {\bf Decomposition of $  \Pi^{\varepsilon, \eta}(t_0)(\mathfrak{U}_{t,t_0}^{\varepsilon, \eta}\Ab{A}) - \Pi^{\varepsilon, \eta}(t)(A) $}\\[1mm]
As a first step  we construct  a two-parameter family of automorphisms $\beta_n^{\varepsilon, \eta,  \Lambda_l, \Lambda_k}(t)$ where the unperturbed Hamiltonian $H_0(t)$ and the ``perturbations'' $V_v$, $H_1$ and $\dot{H}_0$ enter on different scales $\Lambda_l$ and $\Lambda_k$ with $  l\geq k$.

  To do so, we  extend each restriction $v_{\infty}^{\Lambda_k}:\Lambda_k\to \R$ of the infinite volume Lipschitz potential $v_\infty$ to a function  $\tilde v^{\Lambda_k}:\Gamma\to\R$ in such a way  that $\mathrm{supp}( \tilde v^{\Lambda_k}) \subset \Lambda_{2k}$
and $\tilde v^{\Lambda_k}$ satisfies the  Lipschitz condition \eqref{C_v} on all of $\Gamma$, possibly with a different but $k$-independent constant ${C}^\circ_{\tilde v}$. Hence, $\tilde v^{\Lambda_k}$ defines a Lipschitz potential on any $\Lambda_l$ with $l\geq 2k$.
Moreover, the corresponding operator satisfies
$V_{\tilde v}^{\Lambda_k} \in \mathcal{A}_{\Lambda_{l}} \subset \Aloc \subset \mathcal{D}_f$ for any $l\geq 2k$ and all $f \in \mathcal{S}$.

By $\mathfrak{U}_{t,t_0}^{\varepsilon, \eta,\Lambda_l, \Lambda_k}$
we denote the dynamics generated by $H_0^{\Lambda_l} + \varepsilon (V^{\Lambda_k}_{ \tilde v} + H_1^{\Lambda_k})   =: H_0^{\Lambda_l} + \epsi V^{\Lambda_k}$  on $\mathcal{A}_{\Lambda_l}$  with adiabatic parameter $\eta>0$.  Similarly, let $\beta_n^{\varepsilon, \eta,  \Lambda_l, \Lambda_k}(t)$ denote the automorphism on $\mathcal{A}_{\Lambda_l}$   constructed from the inputs $H_0^{\Lambda_l}$, $ V_{\tilde v}^{\Lambda_k}$, $H_1^{\Lambda_k}$  and $\dot{H}_0^{\Lambda_k}$ (cf.~\eqref{A1}), 
and  define the state
\begin{equation*}
\Pi_n^{\varepsilon, \eta, \Lambda_l, \Lambda_k}(t) := \rho_0(t) \circ \beta_n^{\varepsilon, \eta,  \Lambda_l, \Lambda_k}(t) 
\end{equation*}
  on $\mathcal{A}_{\Lambda_l}$. 
In the following we will    consider $H^{\epsi,k} := \{ H^{\Lambda_l}_0 + \epsi H_1^{\Lambda_k} \}_{l\in\N}$ for each fixed $k\in\N$ as an SLT-operator and use, often implicitly, that the   $\Vert\cdot \Vert_{I, \zeta, n }$-norms of the corresponding interactions are bounded uniformly in $k$.
  
From now on, we drop the superscripts $\varepsilon$ and $\eta$, and fix $X\in \mathcal{P}_0(\Gamma)$ and $A\in \mathcal{A}_X$. Then for any $l,k\in\N$ with $X\subset \Lambda_l$ and $l\geq 2k$ we can  
split the left hand side of \eqref{adiabboundbulk} as
\begin{align}
\left\vert \Pi_n^{}(t_0)(\mathfrak{U}_{t,t_0}\Ab{A} )-\Pi_n^{}(t)(A)\right\vert \label{first}
& \;\le\;  \left\vert \Pi_n^{}(t_0)(\mathfrak{U}_{t,t_0}\Ab{A}) -\Pi_n^{ \Lambda_l, \Lambda_k}(t_0)(\mathfrak{U}_{t,t_0}^{ \Lambda_l, \Lambda_k}\Ab{A})  \right\vert\\ 
 \label{fifth}
& \qquad+ \ \left\vert \Pi_n^{ \Lambda_l, \Lambda_k}( t_0)(\mathfrak{U}_{t,t_0}^{ \Lambda_l, \Lambda_k}\Ab{A}) - \Pi_n^{ \Lambda_l, \Lambda_k}( t)(A)\right\vert 
\\ \label{second}
& \qquad+ \ \left\vert \Pi^{\Lambda_l, \Lambda_k}_n( t)(A) -\Pi_n^{}( t)(A)\right\vert \,.
\end{align}
   We will now show that \eqref{first} and \eqref{second}  can be made arbitrarily small by choosing  $l,k$ large enough. To be precise, by this phrase we mean that for any $\delta>0$ there exists $N_\delta\in \N$ such that the corresponding term is smaller than $\delta$ for all $k\geq N_\delta$ and all $l\geq 2k$.   For~\eqref{fifth} we show that it satisfies the bound \eqref{adiabboundbulk},  up to terms, that vanish in the limit $l\to \infty$ for any $k\in\N$. Together these statements then imply \eqref{adiabboundbulk}.\\[3mm]
\noindent  \textbf{(b)} {\bf Thermodynamic limits of the automorphisms}\\[1mm] 
To see that \eqref{first} and \eqref{second} are small for $l,k$ large enough, we merely use that the involved automorphisms have a thermodynamic limit. We only discuss \eqref{first} in some detail and estimate
	\begin{align}
\left\vert \Pi_n^{}(t_0)(\mathfrak{U}_{t,t_0}^{}\Ab{A}) \; -\right. & \left. \Pi_n^{ \Lambda_l, \Lambda_k}(t_0)(\mathfrak{U}_{t,t_0}^{ \Lambda_l, \Lambda_k}\Ab{A})  \right\vert \nonumber \\
\le &\left\vert \Pi_n^{}(t_0)(\mathfrak{U}_{t,t_0}^{}\Ab{A}) -\Pi_n^{ \Lambda_l, \Lambda_l}(t_0)(\mathfrak{U}_{t,t_0}^{ \Lambda_l, \Lambda_l}\Ab{A})  \right\vert \label{est1}\\
&+ \left\vert \Pi_n^{ \Lambda_l, \Lambda_l}(t_0)(\mathfrak{U}_{t,t_0}^{ \Lambda_l, \Lambda_l}\Ab{A}) - \Pi_n^{ \Lambda_l, \Lambda_k}(t_0)(\mathfrak{U}_{t,t_0}^{ \Lambda_l, \Lambda_k}\Ab{A}) \right\vert\, . \label{est2}
	\end{align}
	The first term \eqref{est1} can be made arbitrarily small by choosing $l$ large enough. The precise argument  was carried out  in the proof of Theorem 3.5 in \cite{henheikteufel20202} and is not repeated here. For the second term, one combines  Theorem~3.4 and Theorem~3.8 from \cite{nachtergaele2019quasi} (see also Theorem~B.2 from \cite{henheikteufel20202}) in the same manner as in Proposition~\ref{lrb}. Using these statements and the construction procedure (see \eqref{A1} and \eqref{Aj}) of $\varepsilon {S_n^{}} = \sum_{j = 1}^{n} \varepsilon^{j} {A_j^{}}$ generating the automorphism $\beta_n^{\Lambda_l, \Lambda_k}$, one obtains that \eqref{est2} can be made arbitrarily small by choosing $l,k$ large enough. Note, that the composition of the automorphisms can be estimated with the aid of Lemma 4.1 from \cite{henheikteufel20202}. One can deal with \eqref{second} similarly.\\[3mm]
\textbf{(c)} {\bf Thermodynamic limit of the adiabatic approximation}\\[1mm]We are now left to study the term \eqref{fifth} that compares the full time evolution with  the adiabatic evolution on the cube $\Lambda_l$ with the perturbations restricted to the smaller cube $\Lambda_k$. 
For this difference one   obtains an expansion in powers of $\varepsilon$ exactly like in the proof of Proposition~5.1 in \cite{teufel2020non}. However, in the absence of a spectral gap of the finite volume Hamiltonians, the coefficients up to order $n$ are not identically zero anymore. Instead we need to show 
 that they  vanish as $l \to \infty$ for arbitrary $k \in \mathbb{N}$. To do so, we will apply Proposition~\ref{invliou} and the technical lemmata in Appendix~\ref{technicallemmata}, which have been summarised in Proposition~\ref{prop:rapidtdl}.

  To begin, observe that for  $A \in \mathcal{A}_{\mathrm{loc}}$ 
\begin{align}\label{erstabsch}\nonumber
  \Pi_n^{ \Lambda_l, \Lambda_k}(t_0)(\mathfrak{U}_{t,t_0}^{ \Lambda_l, \Lambda_k}\Ab{A}) &- \Pi_n^{\varepsilon,\eta, \Lambda_l, \Lambda_k}(t)(A) \\
= \  & \int_{t_0}^{t} \mathrm{d}s \ \frac{\mathrm{d}}{\mathrm{d}s} \,\rho_0(s)  \left(  \beta_n^{ \Lambda_l, \Lambda_k}(s)\circ  \mathfrak{U}_{t,s}^{ \Lambda_l, \Lambda_k}\Ab{A} \right)  \,.
\end{align}
According to  Lemma~\ref{quasilocaltimeevolution1} and Lemma~\ref{quasilocalautomorphism} there exists a positive constant $C_\eta$ such that
\begin{equation} \label{firstestimate}
\sup_{k,l \in \mathbb{N}} \left\|\beta_n^{ \Lambda_l, \Lambda_k} (s) \circ \mathfrak{U}_{t,s}^{ \Lambda_l, \Lambda_k}\Ab{A}  \right\Vert_{ f_2} \le C_\eta \Vert A \Vert_{f_1}
\end{equation}
uniformly for $s, t$ in compacts and for suitable $f_1$ and $f_2$. The
 bound in \eqref{firstestimate} is uniform in $l$ and $k$, since the constants on the right hand side of the relevant estimates in Lemma~\ref{quasilocaltimeevolution1} and Lemma~\ref{quasilocalautomorphism} only depend on the interaction norm, $\Vert \cdot \Vert_{I, \zeta, n }$, and the constant, $  C^\circ_{\tilde v}$, of the Lipschitz potential, which are both independent of $l$ and $k$. However, the constant $C_\eta$ diverges as $\eta$ goes to zero. But this will be of no concern for us, since we will use  \eqref{firstestimate} only for showing that certain quantities vanish in the thermodynamic limit at a fixed value of $\eta>0$.

At this point, a technical digression is necessary. In \eqref{firstestimate} and many times in the following we will consider    (sequences of) automorphisms  or derivations on the algebra $\mathcal{A}$ as (sequences of) bounded operators between normed spaces $(\mathcal{D}_{f_1}, \|\cdot\|_{f_1})$ and $(\mathcal{D}_{f_2}, \|\cdot\|_{f_2})$ for ``suitable'' $f_1,f_2\in\mathcal{S}$, where $f_2$ dominates  $f_1$, i.e.\ $f_1\prec f_2$. The precise connection between $f_1$ and $f_2$ is quantified by the corresponding lemmata in Appendix~\ref{technicallemmata} and also depends on the weight function $\zeta$ of the generator of the automorphism resp.\ derivation. In Lemma~\ref{SeqLem} we show that relative to any fixed weight function $\zeta \in\mathcal{S}$ there exists a sequence $(f_j)_{j\in\N}$ of weight functions $f_j\in\mathcal{S}$ such that the pair $f_i,f_{i+1}$ is ``suitable'' for all $i\in\N$, i.e.\ that it can be used in all the lemmata of Appendix~\ref{technicallemmata} together with any decay function $\xi\prec \zeta$.
 Now in this proof we fix such a sequence $(f_j)_{j\in\N}$ relative to the ``worst'' weight function $\zeta_n$ appearing in the  construction at level~$n$, i.e.\  all (finitely many) SLT-operators  are elements of $\mathcal{L}^\circ_{I, \zeta_n, \infty}$. 
More precisely, let $(f_j)_{j\in\N}$ be the sequence in $\mathcal{S}^R_\zeta$ constructed in Lemma~\ref{SeqLem} 
for $\zeta(N)  = \zeta_{n}\left(\lfloor \frac{N}{2} \rfloor^\gamma \right)^\frac12$ and  $R   = C \cdot 2^n$ with some $C < \infty$ (a polynomial growth would also suffice, as follows from the discussion in the proof of  Lemma~\ref{invliouconvlemma}). 
Hence,  in the finitely many norm estimates like \eqref{firstestimate} that will appear in the proof, we can just pick successive pairs $f_i,f_{i+1}$ from this sequence $f_1\prec f_2\prec f_3 \cdots \prec f_i \prec f_{i+1}$. 
As the precise weight function of the final target space plays no role in the statement of the theorem, we do not keep track of this explicitly. 

We now return to \eqref{erstabsch}.
Assuming (G3), the derivative can be evaluated by using the linearity and continuity of the involved maps. We find
\begin{align}
 \frac{\mathrm{d}}{\mathrm{d}s} \,  \rho_0(s) \left( \beta_n^{ \Lambda_l, \Lambda_k}(s)\circ  \mathfrak{U}_{t,s}^{ \Lambda_l, \Lambda_k}\Ab{A}  \right)  
= \ &\left(\frac{\mathrm{d}}{\mathrm{d}s}\,\rho_0(s) \right) \left( \beta_n^{ \Lambda_l, \Lambda_k}(s)\circ  \mathfrak{U}_{t,s}^{ \Lambda_l, \Lambda_k}\Ab{A} \right)\label{deco1}\\
&+ \  \rho_0(s)  \left( \left(\frac{\mathrm{d}}{\mathrm{d}s} \,\beta_n^{ \Lambda_l, \Lambda_k}(s)\right)\circ  \mathfrak{U}_{t,s}^{ \Lambda_l, \Lambda_k}\Ab{A}\right) \label{deco2}\\
&+ \  \rho_0(s)  \left( \beta_n^{ \Lambda_l, \Lambda_k}(s)\circ  \left(\frac{\mathrm{d}}{\mathrm{d}s}\,\mathfrak{U}_{t,s}^{ \Lambda_l, \Lambda_k}\right)\Ab{A}\right) .\label{deco3}
\end{align}
The first term \eqref{deco1} can be evaluated with the help of the following lemma, cf.\ Equation~(2.27) in \cite{moon2019automorphic}. Its proof is given after the proof of the theorem. 
\begin{lem}{\rm (Derivative of the ground state) } \label{derivativeofstate} ~\\
Let $f \in \mathcal{S}$ and $A \in \mathcal{D}_f$. Then  
\begin{equation*} 
\dot{\rho_0}(s)(A) = \mathrm{i} \, \int_{\mathbb{R}} \mathrm{d}v \,w(v) \int_{0}^{v} \mathrm{d}u \, \rho_0(s)  \left(\mathcal{L}_{\dot{H}_0(s)} \circ \mathrm{e}^{\mathrm{i}u\mathcal{L}_{H_0(s)}}\Ab{A}\right).
\end{equation*}
\end{lem}
\noindent Proceeding with \eqref{deco1} we obtain
\begin{align}
 \left(\frac{\mathrm{d}}{\mathrm{d}s}\rho_0(s) \right) &\left( \beta_n^{ \Lambda_l, \Lambda_k}(s)\circ  \mathfrak{U}_{t,s}^{ \Lambda_l, \Lambda_k}\Ab{A}\right) \label{firstterm}\\
= \ &\mathrm{i} \, \int_{\mathbb{R}} \mathrm{d}v \,w(v) \int_{0}^{v} \mathrm{d}u \, \rho_0(s)  \left( \mathcal{L}_{\dot{H}_0(s)} \circ \mathrm{e}^{\mathrm{i}u\mathcal{L}_{H_0(s)}} \circ \beta_n^{ \Lambda_l, \Lambda_k}(s)\circ  \mathfrak{U}_{t,s}^{ \Lambda_l, \Lambda_k}\Ab{A}\right)\,. \nonumber
\end{align}
For \eqref{deco2} we get
\begin{align}
 \rho_0(s) \;  & \left(\left(\frac{\mathrm{d}}{\mathrm{d}s} \beta_n^{ \Lambda_l, \Lambda_k}(s)\right)\circ  \mathfrak{U}_{t,s}^{ \Lambda_l, \Lambda_k}\Ab{A}\right) \label{secondterm}\\
= \ &- \mathrm{i} \,\rho_0(s)  \left(   \left[   \int_{0}^{1} \mathrm{d} \lambda \,\mathrm{e}^{-\mathrm{i} \lambda \varepsilon S_n^{ \Lambda_l, \Lambda_k}(s)} \varepsilon \dot{S}_n^{ \Lambda_l, \Lambda_k}(s) \mathrm{e}^{ \mathrm{i} \lambda \varepsilon S_n^{ \Lambda_l, \Lambda_k}(s)} , \beta_n^{ \Lambda_l, \Lambda_k}(s)\circ  \mathfrak{U}_{t,s}^{ \Lambda_l, \Lambda_k}\Ab{A}\right]\right)\,, \nonumber
\end{align}
 and  for \eqref{deco3}
\begin{align}
 \rho_0(s)     &  \left(\beta_n^{ \Lambda_l, \Lambda_k}(s)\circ  \left(\frac{\mathrm{d}}{\mathrm{d}s}\mathfrak{U}_{t,s}^{ \Lambda_l, \Lambda_k}\right)\Ab{A}\right) \label{thirdterm}\\ 
= \ &-\frac{\mathrm{i}}{\eta}\, \rho_0(s) \left(  \left[\mathrm{e}^{-\mathrm{i}\varepsilon S_n^{ \Lambda_l, \Lambda_k}(s)} \left( H_0^{\Lambda_l}(s) + \varepsilon V^{\Lambda_k}(s) \right)\mathrm{e}^{\mathrm{i}\varepsilon S_n^{\Lambda_l, \Lambda_k}(s)},\beta_n^{\Lambda_l, \Lambda_k}(s)  \circ \mathfrak{U}_{t,s}^{ \Lambda_l, \Lambda_k}\Ab{A}\right]\right) \,.\nonumber
\end{align}

\noindent \textbf{(i)} {\bf Approximating \eqref{firstterm}}:  We now replace \eqref{firstterm} by an expression involving the map~$\mathcal{I}$  defined in Appendix~\ref{invliouappendix}. Using that the function $w$ appearing  in its definition is even and that $\rho_0(s)$ is a  $ \mathcal{L}_{H_0(s)}$-ground state, we have 
\begin{align}
 \int_{\mathbb{R}} \mathrm{d}v \,&w(v) \int_{0}^{v} \mathrm{d}u \, \rho_0(s)  \left( \mathcal{L}_{\dot{H}_0(s)} \circ \mathrm{e}^{\mathrm{i}u\mathcal{L}_{H_0(s)}} \circ \beta_n^{ \Lambda_l, \Lambda_k}(s)\circ  \mathfrak{U}_{t,s}^{ \Lambda_l, \Lambda_k}\Ab{A} \right)\nonumber \\
&+ \ \rho_0(s) \left(\left[\mathcal{I}_s^{\Lambda_l}(\dot{H}_0^{\Lambda_k}(s)), \beta_n^{ \Lambda_l, \Lambda_k}(s)\circ  \mathfrak{U}_{t,s}^{ \Lambda_l, \Lambda_k}\Ab{A}\right]\right) \label{replace} \\ 
= \ & \int_{\mathbb{R}} \mathrm{d}v \, w(v) \int_{0}^{v} \mathrm{d}u \, \rho_0(s) \left(\mathrm{e}^{- \mathrm{i}u\mathcal{L}_{H_0(s)}} \circ \mathcal{L}_{\dot{H}_0(s)} \circ \mathrm{e}^{\mathrm{i}u\mathcal{L}_{H_0(s)}} \circ \beta_n^{ \Lambda_l, \Lambda_k}(s)\circ  \mathfrak{U}_{t,s}^{ \Lambda_l, \Lambda_k}\Ab{A}\right) \nonumber\\ 
&- \  \int_{\mathbb{R}} \mathrm{d}v \, w(v) \int_{0}^{v}\mathrm{d}u \,  \rho_0(s) \left(\mathrm{e}^{- \mathrm{i}u\mathcal{L}^{\Lambda_l}_{H_0(s)}} \circ \mathcal{L}^{\Lambda_k}_{\dot{H}_0(s)} \circ \mathrm{e}^{\mathrm{i}u\mathcal{L}^{\Lambda_l}_{H_0(s)}} \circ \beta_n^{ \Lambda_l, \Lambda_k}(s)\circ  \mathfrak{U}_{t,s}^{ \Lambda_l, \Lambda_k}\Ab{A}\right)\,.\nonumber
\end{align}
Using \eqref{firstestimate} together with Lemma~\ref{convofderivlemma} and Lemma~\ref{quasilocaltimeevolution2}, we see that this difference can be made arbitrarily small for $l,k$ large enough. More precisely, by the uniform bound in  \eqref{firstestimate} (see also Lemma~\ref{quasilocaltimeevolution1} and Lemma~\ref{quasilocalautomorphism}), we have the convergence of the dynamics, Lemma~\ref{quasilocaltimeevolution2}, as well as the convergence of the derivation, Lemma~\ref{convofderivlemma}, both in suitable $f$-norms. Hence, we may replace \eqref{firstterm} by the negative of \eqref{replace} up to an arbitrarily small error for $l,k$ large enough.\\[2mm]
\noindent \textbf{(ii)} {\bf Terms up to  order $n$:} In the second step, we expand the first entry of the commutator in \eqref{secondterm} and \eqref{thirdterm} together with the negative of \eqref{replace}  as in Proposition~5.1 of~\cite{teufel2020non}. Having $-\mathrm{i}/\eta$ as a common prefactor, we get 
\begin{align*}
 \eta \int_{0}^{1} \mathrm{d} \lambda \,& \mathrm{e}^{-\mathrm{i} \lambda \varepsilon S_n^{ \Lambda_l, \Lambda_k}} \varepsilon \dot{S}_n^{ \Lambda_l, \Lambda_k} \, \mathrm{e}^{ \mathrm{i} \lambda \varepsilon S_n^{ \Lambda_l, \Lambda_k}}  +\ \mathrm{e}^{-\mathrm{i}\varepsilon S_n^{ \Lambda_l, \Lambda_k}} \left( H_0^{\Lambda_l} + \varepsilon V^{\Lambda_k} \right)\mathrm{e}^{\mathrm{i}\varepsilon S_n^{\Lambda_l, \Lambda_k}}  + \ \eta\,\mathcal{I}^{\Lambda_l}(\dot{H}_0^{\Lambda_k}) \\
=:  \; &H_0^{\Lambda_l} + \sum_{j=1}^{n} \varepsilon^j R_j^{\Lambda_l, \Lambda_k} + \varepsilon^{n+1} R_{n+1}^{\Lambda_l, \Lambda_k}.
\end{align*}
Inserting the leading order term $H_0^{\Lambda_l}$ back into the commutator in  \eqref{thirdterm}, we find that 
\begin{align*}
\lim\limits_{l \to \infty}\rho_0(s) &\left(\left[H_0^{\Lambda_l}(s), \beta_n^{ \Lambda_l, \Lambda_k}(s)  \circ \mathfrak{U}_{t,s}^{ \Lambda_l, \Lambda_k}\Ab{A} \right]\right)\\ & =  \lim\limits_{l \to \infty}
\rho_0(s)\left( \mathcal{L}_{H_0^{\Lambda_l}(s)} \circ \beta_n^{ \Lambda_l, \Lambda_k}(s)  \circ \mathfrak{U}_{t,s}^{ \Lambda_l, \Lambda_k}\Ab{A} \right) = 0 
\end{align*}
 for any $k \in \mathbb{N}$ and uniformly for $s$ and $t$ in compacts.
Here we used   \eqref{firstestimate} and the fact that Lemma~\ref{convofderivlemma} and   continuity of $\rho_0(s)$ imply  that
\[
 \lim_{l \to \infty} \rho_0(s) \circ \mathcal{L}_{H_0^{\Lambda_l}(s)} =  \rho_0(s) \circ \mathcal{L}_{H_0 (s)} \equiv 0
\]
uniformly on bounded subsets of $\mathcal{D}_f$.

The first order term has the form 
\begin{equation*}
R_1^{\Lambda_l, \Lambda_k}= -\mathrm{i} \mathcal{L}_{H_0}^{\Lambda_l} (A_1^{\Lambda_l, \Lambda_k}) + \left(\tfrac{\eta}{\varepsilon}\,\mathcal{I}^{\Lambda_l}(\dot{H}_0^{\Lambda_k}) -  V^{\Lambda_k} \right) =: -\mathrm{i} \mathcal{L}_{H_0}^{\Lambda_l} (A_1^{\Lambda_l, \Lambda_k}) +\tilde{R}_1^{\Lambda_l, \Lambda_k}  
\end{equation*}
with
\begin{equation}
A_1^{\Lambda_l, \Lambda_k} := \mathcal{I}^{\Lambda_l}\left(V^{\Lambda_k} - \tfrac{\eta}{\varepsilon}\,\mathcal{I}^{\Lambda_l}(\dot{H}_0^{\Lambda_k})\right). \label{A1}
\end{equation} 
With the aid of Lemma~\ref{invliouconvlemma}, it follows  that, for any $k \in \mathbb{N}$, $A_1^{\Lambda_l, \Lambda_k}$ is convergent in a suitable $f$-norm as $l \to \infty$. So, using Lemma~\ref{convofderivlemma} and Lemma~\ref{invliouconvlemma}, for any $k \in \mathbb{N}$, also $R_1^{\Lambda_l, \Lambda_k}$ is convergent in a suitable $f$-norm and its limit is given by
\begin{equation*}
- \mathrm{i} \mathcal{L}_{H_0} \left( \mathcal{I}\left( \tfrac{\eta}{\varepsilon}\,\mathcal{I}(\dot{H}_0^{\Lambda_k})- V^{\Lambda_k}\right)\right) + \left(\tfrac{\eta}{\varepsilon}\,\mathcal{I}(\dot{H}_0^{\Lambda_k}) -  V^{\Lambda_k} \right). 
\end{equation*}
Thus, we get by using \eqref{firstestimate} and Proposition~\ref{invliou} that
\begin{equation*}
\lim\limits_{l \to \infty }\rho_0(s)\left(\left[R_1^{\Lambda_l, \Lambda_k}(s), \beta_n^{ \Lambda_l, \Lambda_k}(s)  \circ \mathfrak{U}_{t,s}^{ \Lambda_l, \Lambda_k}\Ab{A}\right]\right) = 0
\end{equation*}
for any $k \in \mathbb{N}$ and uniformly for $s$ and $t$ in compacts.

The remainder terms $R_j^{\Lambda_l, \Lambda_k}(s)$ for $j=2,...,n$ have the form 
\begin{equation*}
R_j^{\Lambda_l, \Lambda_k} = -\mathrm{i} \mathcal{L}_{H_0}^{\Lambda_l} (A_j^{\Lambda_l, \Lambda_k}) + \tilde{R}_j^{\Lambda_l, \Lambda_k}
\end{equation*}
where the operators $A_j^{\Lambda_l, \Lambda_k}$ are determined inductively as
\begin{equation} 
A_j^{\Lambda_l, \Lambda_k} := - \mathcal{I}^{\Lambda_l}(\tilde{R}_j^{\Lambda_l, \Lambda_k}) \label{Aj}
\end{equation}
and $\tilde{R}_j^{\Lambda_l, \Lambda_k}$ is composed of a finite number of iterated commutators of the operators $A_i^{\Lambda_l, \Lambda_k}$ and $\dot{A}_i^{\Lambda_l, \Lambda_k}$, $i=1,..., j-1$, with $H_0^{\Lambda_l}$ and $V^{\Lambda_k}$.\footnote{Again, as shown in \cite{teufel2020non}, $A_j^{}$ is a polynomial in $\frac{\eta}{\varepsilon}$ of degree $j$ with coefficients in $\mathcal{L}_{I,\mathcal{S}, \infty}$ for every $j \in \mathbb{N}$.}
By this structure, we have with the aid of Lemma~\ref{convofderivlemma}, Lemma~\ref{invliouconvlemma} and  the continuity of multiplication in $(\mathcal{D}_f,\|\cdot\|_f)$, that
$R_j^{\Lambda_l, \Lambda_k}$ is convergent in a suitable $f$-norm and we get 
\begin{equation*}
\lim\limits_{l \to \infty }\rho_0(s)\left(\left[R_j^{\Lambda_l, \Lambda_k}(s), \beta_n^{ \Lambda_l, \Lambda_k}(s)  \circ \mathfrak{U}_{t,s}^{ \Lambda_l, \Lambda_k}\Ab{A}\right]\right) = 0
\end{equation*}
for $j=2,...,n$, any $k \in \mathbb{N}$ and uniformly for $s$ and $t$ in compacts by Proposition~\ref{invliou} and \eqref{firstestimate}. So, the coefficients up to order  $\varepsilon^n$ vanish in the limit $l \to \infty$. \\[2mm]
\textbf{(iii) Remainder term:} In the last step, the remainder term involving $R_{n+1}^{\Lambda_l, \Lambda_k}$ can be estimated as in proof of Proposition~5.1 and Theorem~5.1 in \cite{teufel2020non} by
\begin{align*}
 &\hspace{-45pt}\frac{\varepsilon^{n+1}}{\eta}\left\vert \int_{t_0}^{t} \mathrm{d}s \ \rho_0(s) \left(  \left[R_{n+1}^{\Lambda_l, \Lambda_k}(s),\beta_n^{ \Lambda_l, \Lambda_k} (s) \circ \mathfrak{U}_{t,s}^{ \Lambda_l, \Lambda_k}\Ab{A}\right] \right)\right\vert \\
\le \ & \frac{\varepsilon^{n+1}}{\eta} \, \vert t-t_0 \vert \sup_{s\in [t_0,t]} \left\Vert   \left[ \left(\beta_n^{ \Lambda_l, \Lambda_k}\right)^{-1} (s)\left(R_{n+1}^{\Lambda_l, \Lambda_k}(s)\right), \mathfrak{U}_{t,s}^{ \Lambda_l, \Lambda_k}\Ab{A}\right]\right\Vert \\ 
\le \ &C_n  \frac{\varepsilon^{n+1}}{\eta} \left(1+ \left(\frac{\eta}{\varepsilon}\right)^{n+1}\right)\vert t -t_0\vert \left(1+ \eta^{-d}\vert t-t_0 \vert^{d}\right) \Vert A \Vert \vert X \vert^2 \\
\le \ & C_n  \frac{\varepsilon^{n+1} + \eta^{n+1}}{\eta^{d+1}} \vert t -t_0\vert \left(1+ \vert t-t_0 \vert^{d}\right) \Vert A \Vert \vert X \vert^2\,,
\end{align*}
where the constant $C_n$ is independent of $l,k \in \mathbb{N}$ and all other parameters, cf.\  \cite{bachmann2018adiabatic,monaco2017adiabatic,teufel2020non}. 
More precisely, the estimate follows from Lemma~C.5 from \cite{monaco2017adiabatic}. Note, that the bound in this Lemma only depends on the the Lieb-Robinson velocity \eqref{lrvelocity}, the relevant interactions norms, and the Lipschitz constant $  C^\circ_{\tilde{v}}$.
\\[3mm]
\textbf{(d)} {\bf Conclusion and resummation}\\[1mm] Summarising our considerations, we have shown that for any $n\in\N$  there exists a positive constant $C_n$ such that
\[
 \left\vert \Pi_n^{\varepsilon, \eta}(t_0)( \mathfrak{U}_{t,t_0}^{\varepsilon, \eta}\Ab{A})-\Pi_n^{\varepsilon, \eta}(t)(A)\right\vert  
\le C_n \,  \frac{\varepsilon^{n+1} + \eta^{n+1}}{\eta^{d+1}} \, \vert t -t_0\vert \left(1+ \vert t-t_0 \vert^{d}\right) \Vert A \Vert \vert X \vert^2\,.
\]
The $n$-independent states $\Pi^{\varepsilon, \eta}$  that satisfy the estimate  \eqref{adiabboundbulk}  can be constructed using Lemma~E.3 and Lemma~E.4 from \cite{henheikteufel20202}. 
All the other statements on $\Pi_n^{\varepsilon, \eta}(t)$ are clear by construction (as in \cite{teufel2020non}). 
\end{proof}
\begin{proof}[Proof of Lemma~\ref{derivativeofstate}]
	The proof of this identity in \cite{moon2019automorphic} uses the technical Lemmata 4.4, 4.5, 4.6, 4.12, 4.13, which we have adapted to our notion of interactions in Appendix~\ref{technicallemmata}, Lemmata~\ref{quasilocaltimeevolution2},~\ref{quasilocaltimeevolution1},~\ref{invliouconvlemma},~\ref{domainofderivlemma}, and~\ref{convofderivlemma}, respectively. Therefore, we only sketch the main arguments. Again, at each step, one chooses suitable $f$-norms.  
	
	Let $f \in \mathcal{S}$ and $A \in \mathcal{D}_f$. 
	First, using the spectral gap of the bulk Hamiltonian, one can show that 
	\begin{equation*}
	\dot{\rho_0}(s)(\mathcal{J}_s\Ab{A}) = 0\,.
	\end{equation*}
	where $\mathcal{J}_s$ is defined in Appendix \ref{invliouappendix}. By application of the Duhamel formula, one arrives at
	\begin{equation*} 
	\dot{\rho_0}(s)(A) = - \mathrm{i} \, \int_{\mathbb{R}} \mathrm{d}v \,w(v) \int_{0}^{v} \mathrm{d}u \, \dot{\rho_0}(s)  \left( \mathcal{L}_{H_0(s)} \circ \mathrm{e}^{\mathrm{i}u\mathcal{L}_{H_0(s)}}\Ab{A}\right).
	\end{equation*}
	Using (I2) in the form of Lemma \ref{domainofderivlemma}, (G3), and that ${\rho_0}(s)  (  \mathcal{L}_{H_0(s)}\Ab{A}) \equiv  0$, since $\rho_0(s)$ is the  $ \mathcal{L}_{H_0(s)}$-ground state, we have
	\begin{equation*}
	\dot{\rho_0}(s) \left(\mathcal{L}_{H_0(s)}\Ab{A}\right) + {\rho_0}(s) \left( \mathcal{L}_{\dot{H_0}(s)}\Ab{A}\right) = 0\,. \qedhere
	\end{equation*}
	 \end{proof}
\subsection{Proof of the adiabatic theorem for finite domains}
\begin{proof}[Proof of Theorem~\ref{existenceofneass4}]
We show that both terms on the left hand side of \eqref{adiabboundfinite} converge to their infinite volume limits with a rate given by the second term on the right hand side of \eqref{adiabboundfinite}. 
Together with Theorem~\ref{existenceofneass3} this implies the claim.
Again, we do this first for fixed  $n$  and comment on the  resummation afterwards. 

Fix $X\in \mathcal{P}_0(\Gamma)$ and $A\in \mathcal{A}_X$.
Then for any $k_2 \le k_1 \le k$ such that $X \subset \Lambda_{k_2}$ we obtain for
for the first term in \eqref{adiabboundfinite} that
\begin{align}
 \left\vert\Pi_n^{\varepsilon, \eta, \Lambda_k} \right.&\left.(t_0)(\mathfrak{U}_{t,t_0}^{\varepsilon, \eta, \Lambda_k}\Ab{A}) - \Pi_n^{\varepsilon, \eta}(t_0)(\mathfrak{U}_{t,t_0}^{\varepsilon, \eta}\Ab{A}) \right\vert \nonumber\\
= \ &\left\vert \rho^{\Lambda_k}_0(t_0)\left(\beta_n^{\varepsilon, \eta, \Lambda_k}(t_0)\circ \mathfrak{U}_{t,t_0}^{\varepsilon, \eta, \Lambda_k}\Ab{A}\right) - \rho^{}_0(t_0)\left(\beta_n^{\varepsilon, \eta}(t_0)\circ \mathfrak{U}_{t,t_0}^{\varepsilon, \eta}\Ab{A}\right) \right\vert\nonumber \\
\le \ &\left\vert \rho^{\Lambda_k}_0(t_0)\left(\beta_n^{\varepsilon, \eta, \Lambda_k}(t_0)\circ \mathfrak{U}_{t,t_0}^{\varepsilon, \eta, \Lambda_k}\Ab{A}\right) - \rho^{\Lambda_k}_0(t_0)\left(\beta_n^{\varepsilon, \eta, \Lambda_k}(t_0)\circ \mathfrak{U}_{t,t_0}^{\varepsilon, \eta, \Lambda_{k_2}}\Ab{A}  \right) \right\vert \nonumber\\
&  +   \left\vert \rho^{\Lambda_k}_0(t_0)\left(\beta_n^{\varepsilon, \eta, \Lambda_k}(t_0)\circ \mathfrak{U}_{t,t_0}^{\varepsilon, \eta, \Lambda_{k_2}}\Ab{A} \right) - \rho^{\Lambda_k}_0(t_0)\left(\beta_n^{\varepsilon, \eta, \Lambda_{k_1}}(t_0)\circ \mathfrak{U}_{t,t_0}^{\varepsilon, \eta, \Lambda_{k_2}}\Ab{A}  \right) \right\vert \nonumber\\
&  +   \left\vert \rho^{\Lambda_k}_0(t_0)\left(\beta_n^{\varepsilon, \eta, \Lambda_{k_1}}(t_0)\circ \mathfrak{U}_{t,t_0}^{\varepsilon, \eta, \Lambda_{k_2}}\Ab{A} \right)- \rho^{}_0(t_0)\left(\beta_n^{\varepsilon, \eta, \Lambda_{k_1}}(t_0)\circ \mathfrak{U}_{t,t_0}^{\varepsilon, \eta, \Lambda_{k_2}}\Ab{A} \right) \right\vert\nonumber\\
& +   \left\vert \rho^{}_0(t_0)\left(\beta_n^{\varepsilon, \eta, \Lambda_{k_1}}(t_0)\circ \mathfrak{U}_{t,t_0}^{\varepsilon, \eta, \Lambda_{k_2}}\Ab{A}\right) - \rho^{}_0(t_0)\left(\beta_n^{\varepsilon, \eta}(t_0)\circ \mathfrak{U}_{t,t_0}^{\varepsilon, \eta, \Lambda_{k_2}}\Ab{A}  \right) \right\vert\nonumber\\
&  +   \left\vert  \rho^{}_0(t_0)\left(\beta_n^{\varepsilon, \eta}(t_0)\circ \mathfrak{U}_{t,t_0}^{\varepsilon, \eta, \Lambda_{k_2}}\Ab{A}\right) - \rho^{}_0(t_0)\left(\beta_n^{\varepsilon, \eta}(t_0)\circ \mathfrak{U}_{t,t_0}^{\varepsilon, \eta}\Ab{A} \right) \right\vert \nonumber\\ 
\le \ & \left\Vert \left(\mathfrak{U}_{t,t_0}^{\varepsilon, \eta, \Lambda_k} -  \mathfrak{U}_{t,t_0}^{\varepsilon, \eta, \Lambda_{k_2}}\right)\Ab{A} \right\Vert \label{1}\\
&  +   \left\Vert  \left(\beta_n^{\varepsilon, \eta, \Lambda_k}(t_0) -   \beta_n^{\varepsilon, \eta, \Lambda_{k_1}}(t_0)\right)\circ \mathfrak{U}_{t,t_0}^{\varepsilon, \eta, \Lambda_{k_2}}\Ab{A}  \right\Vert\label{2} \\
&  +   \left\vert  \left(\rho^{\Lambda_k}_0(t_0) - \rho^{}_0(t_0)\right)\left(   \beta_n^{\varepsilon, \eta, \Lambda_{k_1}}(t_0)\circ \mathfrak{U}_{t,t_0}^{\varepsilon, \eta, \Lambda_{k_2}}\Ab{A} \right) \right\vert \label{3}\\
&  +   \left\Vert  \left(  \beta_n^{\varepsilon, \eta, \Lambda_{k_1}}(t_0) - \beta_n^{\varepsilon, \eta}(t_0)\right)\circ \mathfrak{U}_{t,t_0}^{\varepsilon, \eta, \Lambda_{k_2}}\Ab{A} \right\Vert \label{4}\\
&   +   \left\Vert \left( \mathfrak{U}_{t,t_0}^{\varepsilon, \eta, \Lambda_{k_2}} - \mathfrak{U}_{t,t_0}^{\varepsilon, \eta}\right)\Ab{A} \right\Vert\label{5}\,.
\end{align}
 
For \eqref{1} and \eqref{5}   we combine the estimates provided by Proposition~\ref{lrb} and Corollary~\ref{lrbcorr} with the trivial bound by $2 \Vert A \Vert$ using the following trivial inequality  with $\alpha = \eta$: 
\begin{equation}\label{trivial}
0 \leq c \leq \min(a,b) \qquad\Rightarrow\qquad c \le a^{\alpha} \cdot b^{1-\alpha}  \quad \mbox{for all} \quad \alpha \in (0,1)\,.
\end{equation}
This yields 
\begin{equation*}
\eqref{1} + \eqref{5}\; \le\; C(t,t_0) \, \Vert A \Vert \, \mathrm{diam}(X)^{(d+1)\eta} \exp\left(-a\, \eta \,  \mathrm{dist}(X, \Gamma \setminus \Lambda_{\lceil k_2 - \lambda k_2^{\gamma} \rceil})^{\gamma}\right), 
\end{equation*}
where $C(t,t_0)$   depends only on $t$, $t_0$ and $H_0$.   

Since $\varepsilon S^{\varepsilon, \eta}_n \in \mathcal{L}_{I,\mathcal{S}, \infty, \infty}$ has a rapid thermodynamic limit with exponent $\gamma \in (0,1)$ by application of Lemma~\ref{cauchy1}, Lemma~\ref{cauchy2} and Lemma~\ref{cauchy3}, we have that for some {$\zeta_n \in \mathcal{S}$ in particular $\varepsilon S^{\varepsilon, \eta}_n \in \mathcal{L}_{I,\zeta_n, 0, \infty}$} has a rapid thermodynamic limit with exponent $\gamma \in (0,1)$. Hence, \eqref{2} and \eqref{4} can be estimated using the ``local decomposition technique'',
\eqref{trivial}, Proposition~\ref{lrb} and Corollary~\ref{lrbcorr} as follows:
	For $Z \subset Z'$, let $\mathbb{E}^{Z'}_Z: \mathcal{A}_{Z'} \to \mathcal{A}_Z $ be the conditional expectation (on even observables!). Moreover, for a set $Y \in \mathcal{P}_0(\Gamma)$ and $\delta \ge 0$ let
$ Y_{\delta} = \set{z \in \Gamma : \mathrm{dist}(z,Y) \le \delta}$
be the ``fattening" of the set $Y$ by $\delta$. 
Defining
\begin{equation*}
A^{(0)} := \mathbb{E}^{\Lambda_{k_2}}_{X_{\frac{v}{\eta}\vert t -t_0\vert }\cap \Lambda_{k_2}}\circ \mathfrak{U}_{t,t_0}^{\varepsilon, \Lambda_{k_2}, \eta}\Ab{A}, 
\end{equation*}
where $v$ is the Lieb-Robinson velocity for $\min(a,a')$  as in \eqref{lrvelocity}, and for $j\ge1$
\begin{equation*}
A^{(j)}: =\left( \mathbb{E}^{\Lambda_{k_2}}_{X_{\frac{v}{\eta}\vert t -t_0\vert + j }\cap \Lambda_{k_2}}  - \mathbb{E}^{\Lambda_{k_2}}_{X_{\frac{v}{\eta}\vert t -t_0\vert +j-1}\cap \Lambda_{k_2}}\right)\circ \mathfrak{U}_{t,t_0}^{\varepsilon, \eta, \Lambda_{k_2}}\Ab{A}\,,
\end{equation*}
we can write $\mathfrak{U}_{t,t_0}^{\varepsilon, \Lambda_{k_2}, \eta}\Ab{A} = \sum_{j=0}^{\infty}A^{(j)}$, where the sum is always finite, since eventually $X_{\frac{v}{\eta}\vert t-t_0 \vert + j }\cap \Lambda_{k_2} = \Lambda_{k_2}$. 
Clearly, $A^{(j)} \in \mathcal{A}_{X_{\frac{v}{\eta}\vert t-t_0 \vert + j }}$ and we bound
\begin{equation*}
\mathrm{diam} (X_{\frac{v}{\eta}\vert t-t_0 \vert + j } ) \le C\, \mathrm{diam}(X)\, (j+1) \left(1+\eta^{-1}\vert t-t_0\vert\right). 
\end{equation*}
According to the properties of the conditional expectation on even observables (see Lemma~C.2 from \cite{henheikteufel20202}) and using the Lieb-Robinson bound (Proposition~\ref{lrb}) in combination with \eqref{trivial}, we have
\begin{equation*} 
\Vert A^{(j)} \Vert \le \min\left(C \Vert A \Vert \vert X\vert \mathrm{e}^{-\min(a,a')j}, \, 2 \Vert A\Vert\right) \le C_{\alpha_1} \Vert A \Vert \vert X\vert^{\alpha_1} \mathrm{e}^{-\alpha_1 \min(a,a')j}
\end{equation*}
with $\alpha_1 \in (0,1)$. Thus, with the aid of Corollary~\ref{lrbcorr} and using again \eqref{trivial}, we get
\begin{align*}
&\left\Vert  \left(  \beta_n^{\varepsilon, \eta, \Lambda_{k_1}}(t_0) - \beta_n^{\varepsilon, \eta}(t_0)\right) \circ   \mathfrak{U}_{t,t_0}^{\varepsilon, \eta, \Lambda_{k_2}}\Ab{A}  \right\Vert\; \le \;  \sum_{j=0}^{\infty} \left\Vert  \left(  \beta_n^{\varepsilon, \eta, \Lambda_{k_1}}(t_0) - \beta_n^{\varepsilon, \eta}(t_0)\right)\left\llbracket  A^{(j)} \right\rrbracket \right\Vert \\
&\quad \le \;   C_{\alpha_1, \alpha_2} \Vert A \Vert \,  \mathrm{diam}(X)^{d(\alpha_1+ \alpha_2)+ \alpha_2}(1+\eta^{-1}\vert t-t_0\vert )^{(d+1)\alpha_2} \, \zeta_{\gamma}^{\alpha_2}(\lceil k_1 -  \lambda k_1^{\gamma}\rceil - k_2), 
\end{align*}
where $\alpha_1, \alpha_2 \in (0,1)$, since the sum
\begin{equation*}
\sum_{j=0}^{\infty}   \mathrm{e}^{-\alpha_1 \min(a,a')j} (j+1)^{(d+1)\alpha_2} < \infty
\end{equation*}
is finite for every choice of $\alpha_1, \alpha_2 \in (0,1)$. 
Therefore, we arrive at
\begin{equation*}
\eqref{2} + \eqref{4} \le C_{\alpha_1, \alpha_2} \Vert A \Vert \,  \mathrm{diam}(X)^{(d+1)(\alpha_1+ \alpha_2)}(1+\eta^{-1}\vert t-t_0\vert )^{(d+1)\alpha_2} \, \zeta_{\gamma}^{\alpha_2}(\lceil k_1 -  \lambda k_1^{\gamma}\rceil - k_2).
\end{equation*} 
For \eqref{3}, we perform the local decomposition twice. Using the notation from above, we define 
\begin{equation*}
A^{(j,0)} = \mathbb{E}^{\Lambda_{k_1}}_{X\cap \Lambda_{k_1}}\circ \beta_n^{\varepsilon, \eta, \Lambda_{k_1}}(t_0) \Ab{A^{(j)} } 
\end{equation*}
and for $i \ge 1$
\begin{equation*}
A^{(j,i)} \;= \;\left(\mathbb{E}^{\Lambda_{k_1}}_{X_{i}\cap \Lambda_{k_1}}  - \mathbb{E}^{\Lambda_{k_1}}_{X_{i-1}\cap \Lambda_{k_1}}\right) \circ  \beta_n^{\varepsilon, \eta, \Lambda_{k_1}}(t_0)\Ab{A^{(j)}} \,.
\end{equation*}
Hence we can write 
\begin{equation*}
\beta_n^{\varepsilon, \eta, \Lambda_{k_1}}(t_0)\circ \mathfrak{U}_{t,t_0}^{\varepsilon, \eta, \Lambda_{k_2}}\Ab{A} \;=\;\sum_{j=0}^{\infty} \beta_n^{\varepsilon, \eta, \Lambda_{k_1}}(t_0) \Ab{ A^{(j)}} \;=\; \sum_{i,j = 0}^{\infty} A^{(j,i)}
\end{equation*}
and the sums are always finite. 
Using this form, it can be estimated with the same methods as above by
\begin{equation*}
\eqref{3} \;\le\; C_{\alpha_1, \alpha_2, \alpha_3}\Vert A \Vert \,  \mathrm{diam}(X)^{m\alpha_1+ d(\alpha_2 + \alpha_3)} \, (1+\eta^{-1}\vert t-t_0\vert )^{ d\alpha_2} h^{\alpha_1}(k-k_1), 
\end{equation*}
where $\alpha_1, \alpha_2, \alpha_3 \in (0,1)$. 

Analogously, we estimate the second term on the left hand side of \eqref{adiabboundfinite} by
\begin{align}
 \left\vert\Pi_n^{\varepsilon, \eta, \Lambda_k}(t)(A) - \Pi_n^{\varepsilon, \eta}(t)(A) \right\vert  \;
\le \;& \left\Vert  \left(\beta_n^{\varepsilon, \eta, \Lambda_k}(t) -   \beta_n^{\varepsilon, \eta, \Lambda_{k_3}}(t)\right)\Ab{A} \right\Vert \label{1.1}\\
&+ \ \left\vert  \left(\rho^{\Lambda_k}_0(t) - \rho^{}_0(t)\right)\left(   \beta_n^{\varepsilon, \eta, \Lambda_{k_3}}(t)\Ab{A} \right) \right\vert \label{2.1} \\
&+ \  \left\Vert  \left(\beta_n^{\varepsilon, \eta}(t) -   \beta_n^{\varepsilon, \eta, \Lambda_{k_3}}(t)\right)\Ab{A} \right\Vert \label{3.1}
\end{align}
for $k_3 \le k$ to be chosen and $A \in \mathcal{A}_X$ such that $X \subset \Lambda_{k_3}$. 
As above, we get
\begin{equation*}
\eqref{1.1} + \eqref{3.1} \;\le\; C_{\alpha} \Vert A \Vert \,\mathrm{diam}(X)^{(d+1)\alpha}\, \zeta_{\gamma}^{\alpha}\,(\mathrm{dist}(X, \Gamma \setminus \Lambda_{\lceil k_3 - \lambda k_3^{\gamma}\rceil}))
\end{equation*}
  with $\alpha \in (0,1)$ and
\begin{equation*}
\eqref{2.1}\; \le\; C_{\alpha_1, \alpha_2}\Vert A \Vert \,  \mathrm{diam}(X)^{m\alpha_1+ (d+1)\alpha_2} \,  h^{\alpha_1}(k-k_3), 
\end{equation*}
  with $\alpha_1, \alpha_2 \in (0,1)$. 
Now, we choose
\begin{align*}
k_1 &= \left\lfloor  k-\frac{\mathrm{dist}\left(X,\Gamma \setminus \Lambda\left(k\right)\right) - \lambda k^{\gamma}}{3} \right\rfloor \\
k_2 &= \left\lfloor  k-\frac{2 \left(\mathrm{dist}\left(X,\Gamma \setminus \Lambda\left(k\right)\right) - \lambda k^{\gamma}\right)}{3} \right\rfloor\\
k_3 &= \left\lfloor  k-\frac{\mathrm{dist}\left(X,\Gamma \setminus \Lambda\left(k\right)\right) - \lambda k^{\gamma}}{2} \right\rfloor, 
\end{align*}
which satisfy $k \ge k_1 \ge k_2$ and $k \ge k_3 $, as well as $k_1 \ge \frac{k}{2}$, $k_2 \ge \frac{ k}{4}$ and $k_3 \ge \frac{k}{3}$ for $k$ large enough. 

Putting everything together and choosing the $\alpha$-parameters appropriately, we have shown that
\begin{align*}
&\left\vert \Pi_n^{\varepsilon, \eta, \Lambda_k}(t_0)(\mathfrak{U}_{t,t_0}^{\varepsilon, \eta, \Lambda_k}\Ab{A}) - \Pi_n^{\varepsilon, \eta, \Lambda_k}(t)(A) \right\vert \\
 & \quad \le \; \left\vert \Pi_n^{\varepsilon, \eta}(t_0)(\mathfrak{U}_{t,t_0}^{\varepsilon, \eta}\Ab{A}) - \Pi_n^{\varepsilon, \eta}(t)(A) \right\vert \\
 &\quad\quad + \left\vert\Pi_n^{\varepsilon, \eta, \Lambda_k}(t_0)(\mathfrak{U}_{t,t_0}^{\varepsilon, \eta, \Lambda_k}\Ab{A}) - \Pi_n^{\varepsilon, \eta}(t_0)(\mathfrak{U}_{t,t_0}^{\varepsilon, \eta}\Ab{A}) \right\vert + \left\vert\Pi_n^{\varepsilon, \eta, \Lambda_k}(t)(A) - \Pi_n^{\varepsilon, \eta}(t)(A) \right\vert\\
& \quad \le \;C_n \, \frac{\varepsilon^{n+1} + \eta^{n+1}}{\eta^{d+1}} \, \vert t-t_0\vert \,(1+\vert t- t_0\vert^{d}) \  \Vert A \Vert \, \vert X \vert^2\\
  &\quad\quad + \ \tilde{C}_{n,m,K} \, \left(1+ \eta \,  \mathrm{dist}(X,\Gamma \setminus \Lambda_{\lfloor k-\lambda k^{\gamma} \rfloor } )\right)^{-m}\Vert A \Vert \, \mathrm{diam}(X)^{2d},
\end{align*}
for all $m \in \mathbb{N}$ and compact $K \subset I$, 
which is a valid estimate for all $A \in \mathcal{A}_X$ with $X \subset \Lambda_k$ after a possible adjustment of $\tilde{C}$. 
The $n$-independent states $\Pi^{\varepsilon, \eta, \Lambda_k}(t)$  that satisfy the estimate  \eqref{adiabboundfinite}
  can again be constructed  using Lemma~E.3 and Lemma~E.4 from \cite{henheikteufel20202}.
 All the other statements on $\Pi_n^{\varepsilon, \eta, \Lambda_k}(t)$ are clear by construction (as in \cite{teufel2020non}). 
\end{proof}
\appendix
\section{Inverting the Liouvillian in the bulk} \label{invliouappendix}
In this Appendix, we prove Proposition~\ref{invliou}.
First note that there exists a non-negative, even function $w \in L^1(\mathbb{R})$ with $\Vert w\Vert_1=1$ and
\begin{equation*}
\sup_{s \in \R } \vert s \vert^n \vert w(s) \vert < \infty, \quad \forall n \in \mathbb{N}_0,
\end{equation*} 
such that its Fourier transform 
\begin{equation*}
\hat{w}(k) = \frac{1}{\sqrt{2\pi}} \int_{\mathbb{R}} \mathrm{e}^{-\mathrm{i}ks} \,w(s) \, \mathrm{d}s
\end{equation*}
satisfies $\mathrm{supp}(\hat{w}) \subset [-1,1]$. An explicit function $w$ having all these properties is constructed in \cite{bachmann2012automorphic}. For any $g >0$ (spectral gap of $H_0$), set
\begin{equation*}
w_g(s) = g \,w(gs).
\end{equation*}  
It is clear that $w_g$ is non-negative, even, $L^1$-normalised, and moreover, 
\begin{equation*}
\mathrm{supp}(\hat{w}_g) \subset [-g,g]. 
\end{equation*}
We drop the subscript $g$ from now on. Let $H_0 \in \mathcal{L}_{I, \exp(-a\,\cdot), 0}$   have a rapid thermodynamic limit and define bounded linear maps
\begin{align*}
&\mathcal{J}_t^{\Lambda_k} : \mathcal{A} \to \mathcal{A}, \ A \mapsto \int_{\mathbb{R}} \mathrm{d}s\,w(s) \,\mathrm{e}^{\mathrm{i}s \mathcal{L}_{H_0(t)}^{\Lambda_k}}\Ab{A} \,, \\ 
&\mathcal{I}_t^{\Lambda_k} : \mathcal{A}\to \mathcal{A}, \ A \mapsto \int_{\mathbb{R}}\mathrm{d}s \, w(s) \int_{0}^{s} \mathrm{d}u \, \mathrm{e}^{\mathrm{i}u \mathcal{L}_{H_0(t)}^{\Lambda_k}}\Ab{A},
\end{align*}
for every $s \in I$ and $k \in \mathbb{N}$. The corresponding ``bulk" versions are defined as
\begin{align*}
&\mathcal{J}_t^{} : \mathcal{A} \to \mathcal{A}, \ A \mapsto \int_{\mathbb{R}}\mathrm{d}s\, w(s) \,\mathrm{e}^{\mathrm{i}s \mathcal{L}_{H_0(t)}^{}}\Ab{A} \,, \\ 
&\mathcal{I}_t^{} : \mathcal{A}\to \mathcal{A}, \ A \mapsto \int_{\mathbb{R}}\mathrm{d}s \, w(s) \int_{0}^{s} \mathrm{d}u \, \mathrm{e}^{\mathrm{i}u \mathcal{L}_{H_0(t)}^{}}\Ab{A},
\end{align*}
The integrals can be understood as Bochner integrals on $(\mathcal{A}, \Vert \cdot \Vert )$ by the continuity of $s \mapsto\mathrm{e}^{\mathrm{i}s \mathcal{L}_{H_0(t)} }\Ab{A} $ for any $A \in \mathcal{A}$ (see Proposition~\ref{lrb}). The map $\mathcal{I}_t$ is called the \SLT-inverse of the Liouvillian, which is justified by Proposition~\ref{invliou} and Lemma~\ref{cauchy3}. {See Appendix D of \cite{henheikteufel20202} for a definition of the SLT-inverse of the Liouvillian on interactions.}

\begin{proof}[Proof of Proposition~\ref{invliou}]
	We drop the subscript $H_0$ of $\mathcal{L}$ and $\mathcal{I}$ from Proposition \ref{invliou} to simplify notation. We prove this proposition using the GNS representation of $\rho_0$. Let $P_{\Omega_{\rho_0}} = \ket{\Omega_{\rho_0}} \bra{\Omega_{\rho_0}}$ be the projection on the GNS state and let $\mathcal{I}_{\rho_0}$ be the GNS version of the inverse Liouvillian, i.e.
	\begin{equation*}
	\mathcal{I}_{\rho_0}(A) = \int_{\mathbb{R}} \mathrm{d}s \, w(s) \int_{0}^{s}\mathrm{d}u \, \mathrm{e}^{\mathrm{i} u H_{0,\rho_0}} A \mathrm{e}^{- \mathrm{i} uH_{0,\rho_0}}, \quad A \in \mathcal{B}(\mathcal{H}_{\rho_0}),
	\end{equation*}
	and $\mathcal{J}_{\rho_0}$ analogously, where $H_{0,\rho_0}$ is the GNS Hamiltonian. Then we have
	\begin{equation*}
	\rho_0\left(\mathrm{i}[\mathcal{L}\circ\mathcal{I}\, \Ab{A},B]\right) 
	=   \mathrm{i}\langle \Omega_{\rho_0},[ \pi_{\rho_0}(\mathcal{L}\circ\mathcal{I}\, \Ab{A}),\pi_{\rho_0}(B)] \Omega_{\rho_0} \rangle. 
	\end{equation*}
	Since $\mathcal{I}\Ab{A} \in D(\mathcal{L})$, we have by Theorem~4 from \cite{bratteli1975unbounded} that for all    $B \in D(\mathcal{L})$
	\begin{align*}
	\mathrm{i}\langle \Omega_{\rho_0},[ \pi_{\rho_0}(\mathcal{L}\circ\mathcal{I}\, \Ab{A}),\pi_{\rho_0}(B)] \Omega_{\rho_0} \rangle & \;=\;\mathrm{i}\langle \Omega_{\rho_0},[[H_{0,\rho_0}, \pi_{\rho_0}(\mathcal{I}\Ab{A})] ,\pi_{\rho_0}(B)] \Omega_{\rho_0} \rangle \\
	& \;=\;\mathrm{i}\langle \Omega_{\rho_0},[[H_{0,\rho_0}, \mathcal{I}_{\rho_0}(\pi_{\rho_0}(A))] ,\pi_{\rho_0}(B)] \Omega_{\rho_0} \rangle.
	\end{align*}
	By application of Proposition~6.9 from \cite{nachtergaele2019quasi}, we get
	\begin{equation*}
	\mathrm{i} [H_{0,\rho_0}, \mathcal{I}_{\rho_0}(\pi_{\rho_0}(A))] =  \mathcal{J}_{\rho_0}(\pi_{\rho_0}(A)) - \pi_{\rho_0}(A), 
	\end{equation*}
	i.e.\ we inverted the Liouvillian up to $\mathcal{J}_{\rho_0}(\pi_{\rho_0}(A))$. But, again with the aid of Proposition~6.9 from \cite{nachtergaele2019quasi} and using the cyclicity of the trace, we have 
	\begin{align*}
	\langle \Omega_{\rho_0},[ \mathcal{J}_{\rho_0}(\pi_{\rho_0}(A)),\pi_{\rho_0}(B)] \Omega_{\rho_0} \rangle& \;=\; \mathrm{tr}(P_{\Omega_{\rho_0}}[ \mathcal{J}_{\rho_0}(\pi_{\rho_0}(A)),\pi_{\rho_0}(B)])\\
	& \;=\;\mathrm{tr}(  \pi_{\rho_0}(B)\underbrace{ [P_{\Omega_{\rho_0}}, \mathcal{J}_{\rho_0}(\pi_{\rho_0}(A))]}_{=0}) 
	\;=\;  0\,.
	\end{align*}
	By using the GNS representation again, we also have
	\begin{equation*}
	\langle \Omega_{\rho_0},[\pi_{\rho_0}(A) ,\pi_{\rho_0}(B)]\, \Omega_{\rho_0} \rangle = \rho_0([A,B])\,.
	\end{equation*}
	We thus showed that  for all $B\in D(\mathcal{L})$ 
	\begin{equation*} 
	\rho_0\big(\big[\mathcal{L}\circ\mathcal{I}\, \Ab{A} - \I A,B\big]\big) = 0\,. 
	\end{equation*}
	By density of $D(\mathcal{L})$ in $\mathcal{A}$ and continuity of $\rho_0:\mathcal{A}\to \C$,  the equality holds for all $B\in\mathcal{A}$.
\end{proof}

\section{Quasi-locality estimates} \label{technicallemmata}
In this appendix we show how to control the actions of automorphisms, derivations, and   inverse Liouvillians generated by SLT-operators with a rapid thermodynamic limit on spaces $\mathcal{D}_f\subset \mathcal{A}$ of quasi-local observables of the infinite system.

Let us briefly summarise the structure of and the motivation for the following results. For controlling the adiabatic approximation in the thermodynamic limit  we need to consider the actions of automorphisms, derivations and   inverse Liouvillians also on the spaces $\mathcal{D}_f$   introduced in \cite{moon2019automorphic}. Parts of the statements of the  required lemmata, Lemmata~\ref{quasilocaltimeevolution1},   \ref{quasilocaltimeevolution2},   \ref{quasilocalautomorphism},   \ref{domainofderivlemma},   \ref{convofderivlemma},  and    \ref{invliouconvlemma}, were established already in \cite{moon2019automorphic}, however,   without explicit uniformity and  only for SLT-operators defined by restrictions of a fixed interaction   on~$\mathcal{P}_0(\Gamma)$.  
The SLT-operators $A_j$ appearing in the adiabatic expansion are not of this form, even if we would assume this form for our original Hamiltonian. 
Thus we generalise the results of \cite{moon2019automorphic} to SLT-operators having a thermodynamic limit and make the uniformity explicit. 
In addition, we also prove norm convergence of automorphisms and derivations in spaces of bounded operators from $(\mathcal{D}_{f_1},\|\cdot\|_{f_1})$ to $(\mathcal{D}_{f_2},\|\cdot\|_{f_2})$ in the thermodynamic limit. For this, however,  the quantitative notion of rapid thermodynamic limit is required. 

As   starting points, we first establish the convergence of automorphisms in the thermodynamic limit as in \cite{nachtergaele2019quasi}, however,  with additional quantitative control on the rate of convergence implied by the condition that the generator has a rapid thermodynamic limit (Proposition~\ref{lrb} and Corollary~\ref{lrbcorr}). Similarly, we prove quantitative estimates on the rate of convergence of derivations   in the thermodynamic limit (Proposition~\ref{tdlofderivations} and Corollary~\ref{tdlofderivationscorr}).

\subsection{Dynamics}

Since we refer to the Lieb-Robinson bounds    several times in this work, we restate them for convenience of the reader in the following proposition. Its  only novel content is, however,  the  implications on the rate of convergence of automorphism groups   for generators with a rapid thermodynamic limit.
\begin{prop}{\rm (Lieb-Robinson bound and convergence of dynamics)} \label{lrb}
	~\\
	Let $H_0 \in \mathcal{L}_{I, \zeta, 0  }$ and $v \in \mathcal{V}_I$ both have a  thermodynamic limit (see Definition~2.1 in \cite{henheikteufel20202}), and set $H:= H_0 + V_v$. For $A\in \mathcal{A}_{\Lambda_k}$ define the local dynamics as
	\begin{equation*}
	\mathfrak{U}^{\Lambda_k}_{t,s}\Ab{A}  =  U^{\Lambda_k}(t,s)^* \,A\, U^{\Lambda_k}(t,s)
	\end{equation*}
	where $U^{\Lambda_k}(t,s)$ is the solution to the Schrödinger equation 
\[
\mathrm{i }\tfrac{\mathrm{d}}{\mathrm{d}t} U^{\Lambda_k}(t,s) = H^{\Lambda_k} (t) U^{ \Lambda_k}(t,s) 
 \]
with $U^{\Lambda_k}(s,s) = \mathrm{id}$. 
 There exists $C<\infty$ such that    for all $A \in \mathcal{A}_X$ and $B \in \mathcal{A}_Y$ with $X, Y \subset\Lambda_k$ and all $k\in \mathbb{N}$
	\begin{align*}
	&\left\Vert \left[\mathfrak{U}_{t,s}^{\Lambda_k}\Ab{A},B\right] \right\Vert  \le \  C \,\Vert A \Vert  \Vert B \Vert   \min(\vert X\vert, \vert Y \vert ) \,\mathrm{e}^{2C_{\zeta}\vert t-s\vert \Vert\Phi_{H_0} \Vert_{I,\zeta, 0}} \zeta(\mathrm{dist}^{\Lambda_k}(X,Y))\,.
	\end{align*}
	If $H_0 \in \mathcal{L}_{I, \exp(-a \, \cdot), 0 }$ for some $a>0$, one defines the Lieb-Robinson velocity via 
	\begin{equation} \label{lrvelocity}
	v_a := 2 a^{-1}C_{\exp(-a \, \cdot)}\Vert\Phi_{H_0} \Vert_{I,\exp(-a \, \cdot ), 0 },
	\end{equation}
	and obtains the more transparent bound
	\begin{equation*}
	\left\Vert \left[\mathfrak{U}_{t,s}^{\Lambda_k}\Ab{A},B\right] \right\Vert \le C \, \Vert A \Vert  \Vert B \Vert \ \min(\vert X \vert, \vert Y \vert) \ \mathrm{e}^{a(v_a\vert t-s\vert - \mathrm{dist}^{\Lambda_k}(X,Y))}. 
	\end{equation*}
If $H_0$ and $v$ have a rapid thermodynamic limit with exponent $\gamma\in (0,1)$, then
	there exist  $\lambda_1 >0$, $\lambda_2 \in (0,1)$, and $C<\infty$,  such that for all $l,k \in \mathbb{N}$ with $l \ge k$, $X \subset \Lambda_k$ and $A \in \mathcal{A}_X$
		\begin{align}\nonumber
	 \left\Vert (\mathfrak{U}_{t,s}^{\Lambda_l}  - \mathfrak{U}_{t,s}^{\Lambda_k})\Ab{A}\right\Vert &\; \le\; \;  C \,\Vert A\Vert \,  \mathrm{diam}(X)^{d+1}  \,  \mathrm{e}^{2 C_{\zeta} |t-s|\Vert\Phi_{H_0}\Vert_{I, \zeta, 0 }} |t-s| \\
		& \qquad\times     \zeta_{\gamma}\left(\mathrm{dist}^{\Lambda_l}(X, \Lambda_l\setminus  
		\Lambda_{\max\{ \left\lceil k - \lambda_1 k^{\gamma}\right\rceil, \left\lceil\lambda_2 \cdot k\right\rceil\}} )
   \right) 	\,. \label{LRconse}
		\end{align} 
 	In every case above, the   constant $C$   depends only on  $\zeta$,   $\Vert \Phi_{H_0} \Vert_{I, \zeta, 0 }$, and the Lipschitz constant $C_v$.    
\end{prop}
\begin{proof}
	The first part is the standard Lieb-Robinson bound (see \cite{lieb1972finite} for the first proof, for fermionic systems see \cite{bru2016lieb}, \cite{nachtergaele2018lieb}). Invoking the estimate 
	\begin{equation*}
	\sup_{k \in \mathbb{N}} \left\Vert \mathfrak{U}_{t,s}^{\Lambda_k}\Ab{A}- A\right\Vert \, \le \,  
	\left(2  \Vert F_{1}\Vert \Vert\Phi_{H_0}\Vert_{I, \zeta, 0 } + \tfrac{1}{2}C_v  r\right) |t-s| \Vert A \Vert \, \mathrm{diam}(X)^{d+1}
	\end{equation*}
	for any $A \in \mathcal{A}_X$ (see Theorem~3.8 in \cite{nachtergaele2019quasi} and Lemma~\ref{lipschitzcomm}),   the estimate  \eqref{LRconse} with the first alternative in the maximum   is a consequence of Theorem~3.4 and Theorem~3.8 in \cite{nachtergaele2019quasi} by choosing $\Lambda_M = \Lambda_{\lceil k - \lambda_1 k^{\gamma}\rceil}$ with $\lambda_1$ from Definition~\ref{cauchydefinition}  resp.~the alternative characterisation in Lemma \ref{lem:equivalentchar}. The second alternative can easily be concluded from the first and is used in the proofs of the lemmata below. 
\end{proof}

 As a consequence of \eqref{LRconse}, $ (\mathfrak{U}_{t,s}^{\Lambda_k}\Ab{A} )_{k \in \mathbb{N}}$ is a Cauchy sequence in $(\mathcal{A},\|\cdot\|) $ for every  $A \in \mathcal{A}_{\mathrm{loc}}$. By density, the limits   $\mathfrak{U}_{t,s}\Ab{A} := \lim_{k\to\infty}\mathfrak{U}_{t,s}^{\Lambda_k}\Ab{A}$ define   a co-cycle $\mathfrak{U}_{t,s}$ of automorphisms  on $\mathcal{A}$  (the bulk dynamics)  and the map $I\times I \to \mathcal{A}$, $(t,s)\mapsto \mathfrak{U}_{t,s}\Ab{A}$ is continuous for every $A \in \mathcal{A}$.\footnote{A similar argument based on a weaker version of \eqref{LRconse} allows the same conclusion under the weaker assumption of  $H$ having a thermodynamic limit.}

  As the finite volume metrics $d^{\Lambda_k}(\cdot, \cdot)$ are compatible in the bulk, we obtain the corresponding statements for the bulk dynamics.  
\begin{cor}{\rm (Infinite volume dynamics)} \label{lrbcorr} ~\\
Under the conditions of Proposition~\ref{lrb}  there exists $C<\infty$,  such that  for all $X,Y\in\mathcal{P}_0(\Gamma)$, $A \in \mathcal{A}_X$, and $B \in \mathcal{A}_Y$
\begin{align*}
\left\Vert \left[\mathfrak{U}_{t,s}\Ab{A},B\right] \right\Vert  \le \  C\, \Vert A \Vert  \Vert B \Vert   \min(\vert X\vert, \vert Y \vert )\, \mathrm{e}^{2C_{\zeta}\vert t-s\vert \Vert\Phi_{H_0} \Vert_{I,\zeta, 0 }} \zeta(\mathrm{dist}(X,Y)) \,.
\end{align*}
If $H_0$ and $v$ have a rapid thermodynamic limit,  there exist  $\lambda_1 >0$, $\lambda_2 \in (0,1)$ such that 
\begin{align*}
	\left\Vert( \mathfrak{U}_{t,s}  - \mathfrak{U}_{t,s}^{\Lambda_k})\Ab{A}\right\Vert \, \le \,\; &C\, \Vert A\Vert \,  \mathrm{diam}(X)^{d+1}  \,  \mathrm{e}^{2 C_{\zeta} |t-s|\Vert\Phi_{H_0}\Vert_{I, \zeta, 0 }} |t-s| \\
	&\qquad \times \  \zeta_{\gamma}(\mathrm{dist}^{}(X, \Gamma \setminus\Lambda_{\max\{ \left\lceil k - \lambda_1 k^{\gamma}\right\rceil, \left\lceil\lambda_2 \cdot k\right\rceil\}} ) \,.
		\end{align*}
 	In both cases above, the   constant $C$   depends only on  $\zeta$,    $\Vert \Phi_{H_0} \Vert_{I, \zeta, 0 }$, and the Lipschitz constant $C_v$.   
\end{cor}
We are now ready to prove generalisations of Lemma 4.4 and Lemma 4.5 from \cite{moon2019automorphic}.  

\begin{lem}{\rm (Quasi-locality of dynamics I)} \label{quasilocaltimeevolution1} ~\\
	Let $H_0 \in \mathcal{L}_{I, \exp(-a \cdot), 0 }$ and $v \in \mathcal{V}_I$ both have a  thermodynamic limit, set $H=H_0 + V_v$ and let $\mathfrak{U}_{t,s}$ denote the {infinite volume} dynamics generated by $H$. Let $f_1,f_2 :[0,\infty) \to (0,\infty)$ be bounded, non-increasing functions with $\lim_{s\to \infty} f_i(s) = 0$  for $i = 1,2$, such that\footnote{Here, $v_a$ denotes the Lieb-Robinson velocity from \eqref{lrvelocity}. }
	\[
	\int_{ 0}^\infty \mathrm{d}s \, w_g(s) \frac{2s }{f_2(4v_a s)} < \infty\,, 
	\]
	\[
	\sup_{N \in \mathbb{N}}  \frac{f_1( N- \lfloor \frac{N}{2} \rfloor)}{f_2(N)}  < \infty\,,\qquad 
	\sup_{N \in \mathbb{N}}  \frac{N^{d+1} \, \mathrm{e}^{-a \frac{\left\lfloor \frac{N}{2} \right\rfloor}{2}}}{f_2(N)}  < \infty\,,
	\]
	where $w_g$ is defined in Appendix~\ref{invliouappendix}. 
	
 Then $\mathfrak{U}_{t,s} : \mathcal{D}_{f_1} \to \mathcal{D}_{f_2}$ is a bounded operator and the sequence $\mathfrak{U}^{\Lambda_N}_{t,s} : \mathcal{D}_{f_1} \to \mathcal{D}_{f_2}$  of operators is uniformly bounded, both uniformly for $s$ and $t$ in compacts. More precisely, there is a non-negative non-decreasing function $b_{f_1, f_2}: [0,\infty) \to [0,\infty)$ such that
 		\begin{align*}
 	\Vert \mathfrak{U}_{t,s}\Ab{A} \Vert_{f_2} \le b_{f_1, f_2}(\vert t-s \vert) \Vert A \Vert_{f_1} \ \ \text{and} \  \ \sup_{N \in \mathbb{N}} \Vert \mathfrak{U}_{t,s}^{\Lambda_N}\Ab{A} \Vert_{f_2} \le b_{f_1, f_2}(\vert t-s \vert) \Vert A \Vert_{f_1}, 
 	\end{align*}
 	for all $A \in \mathcal{D}_{f_1}$. Moreover, we have that 
	\begin{align*}
	\int_{0}^\infty \mathrm{d}t \, s \,w_g(s) \,  b_{f_1, f_2} (s) < \infty.
		\end{align*}
		Beside the indicated dependence on $f_1$ and $f_2$, the function $b_{f_1,f_2}$ only depends on $a$,   ${\Vert \Phi_{H_0} \Vert_{I,\exp(-a \, \cdot ), 0 }}$, and the Lipschitz constant $C_v$.
\end{lem}
\begin{proof}
	The proof is analogous to the one of Lemma 4.5 from \cite{moon2019automorphic}. 
 Let $A \in \mathcal{D}_{f_1}$. Then, as $\mathfrak{U}_{t,s}$ is an automorphism, 
	\begin{equation*}
	\Vert \mathfrak{U}_{t,s}\Ab{A} \Vert = \Vert A \Vert \le \Vert A\Vert_{f_1}. 
\end{equation*}
From Corollary~\ref{lrbcorr} in combination with Lemma C.2 from \cite{henheikteufel20202}, for $N,k \in \mathbb{N}$ with $k<N$, we obtain
\begin{align*}
\left\Vert ( 1-    \mathbb{E}_{\Lambda_N})\circ \mathfrak{U}_{t,s}\Ab{A}  \right\Vert \; &
\le \;  \left\Vert ( 1 -   \mathbb{E}_{\Lambda_N})\circ \mathfrak{U}_{t,s}\circ \mathbb{E}_{\Lambda_k}\Ab{A}   \right\Vert\; +\; 2\, \Vert (1 - \mathbb{E}_{\Lambda_k})\Ab{A} \Vert \\
  &
\le \; C\, \|A \|\, k^d \,\vert t-s \vert \,\E^{a(v_a \vert t-s \vert-(N-k) )} \;+ \;2 \,\Vert A \Vert_{f_1} f_1(k) . 
\end{align*}
For $N \in \mathbb{N}$ with $4v_a \vert t-s \vert \le N$, we use this bound with $k = N- \left\lfloor \tfrac{N}{2}\right\rfloor$ to estimate
\begin{align*}
\left\Vert(1  -  \mathbb{E}_{\Lambda_N})\circ \mathfrak{U}_{t,s}\Ab{A}  \right\Vert \; &
\le \;   C \, \Vert A \Vert\, N^{d+1} \, \E^{a(v_a \vert t-s \vert-\left\lfloor N/2 \right\rfloor )} \;+\; 2\, \Vert A \Vert_{f_1} f_1\left(N- \left\lfloor \tfrac{N}{2}\right\rfloor\right) \\
&
\le \;   \Big( C N^{d+1}  \E^{-a \frac{\left\lfloor  \frac{N}{2}\right\rfloor}{2}} + 2 f_1\Big(N- \left\lfloor \tfrac{N}{2}\right\rfloor\Big)  \Big) \Vert A \Vert_{f_1}\,.
\end{align*}
On the other hand, for $N \in \mathbb{N}$ with $4v\vert t-s \vert > N$, we simply have
\begin{equation*}
\left\Vert(1-   \mathbb{E}_{\Lambda_N})\circ \mathfrak{U}_{t,s}\Ab{A} \right\Vert\; \le \; 2 \, \Vert A \Vert \;\le\; 2 \,\Vert A \Vert_{f_1}. 
\end{equation*}
Hence, summarising our investigations, we obtain
\begin{align*}
 \Vert \mathfrak{U}_{t,s}\Ab{A} \Vert_{f_2}  \le \; &\left( 1+ \max\left\{   \begin{aligned}
2 \sup_{N \in \mathbb{N}} \left(\tfrac{f_1( N- \lfloor \frac{N}{2} \rfloor)}{f_2(N)}\right) + & \, C
\sup_{N \in \mathbb{N}} \Bigg(\tfrac{N^{d+1} \, \mathrm{e}^{-a \frac{\left\lfloor \frac{N}{2} \right\rfloor}{2}}}{f_2(N)}\Bigg) \, , \\ 
\frac{2 \vert t -s\vert }{f_2(4v_a \vert t-s \vert)} &\chi_{4v_a\vert t-s\vert \ge 1 }
\end{aligned} \right\}\right) \Vert A \Vert_{f_1}  \\
=: & \; b_{f_1,f_2}(\vert t-s\vert) \,\Vert A \Vert_{f_1} \,.
\end{align*}
The claim follows by the requirements on $f_1$ and $f_2$. The uniform boundedness of $\mathfrak{U}^{\Lambda_N}_{t,s}$ can be proven in the same way.
\end{proof}
\begin{lem} {\rm (Quasi-locality of dynamics II)} \label{quasilocaltimeevolution2}~\\
	Let $H_0 \in \mathcal{L}_{I, \exp(-a\, \cdot), 0 }$ and $v \in \mathcal{V}_I$ both have a rapid thermodynamic limit with exponent $\gamma \in (0,1)$, set $H=H_0 + V_v$ and let $\mathfrak{U}_{t,s}$ denote the dynamics generated by $H$. Let $f_1,f_2 :[0,\infty) \to (0,\infty)$ be bounded, non-increasing functions with $\lim_{s\to \infty} f_i(s) = 0$, for $i = 1,2$, such that
	\begin{align*}
	\lim_{N \to \infty}  \frac{f_1(\lfloor \frac{N}{2} \rfloor)}{f_2(N)}  = 0 \qquad\mbox{and}\qquad 
	\lim_{N \to \infty}  \frac{N^{d+1} \, \mathrm{e}^{-a (N-\lfloor \frac{N}{2} \rfloor)^{\gamma}}}{f_2(N)} = 0\,.
	\end{align*}
	Then $\mathfrak{U}_{t,s} : \mathcal{D}_{f_1} \to \mathcal{D}_{f_2}$ and $\mathfrak{U}^{\Lambda_N}_{t,s} : \mathcal{D}_{f_1} \to \mathcal{D}_{f_2}$ are bounded operators, and 
	   $\mathfrak{U}_{t,s}^{\Lambda_N} \xrightarrow{N \to \infty} \mathfrak{U}_{t,s}$
in operator norm, uniformly for $s,t$ in compacts. In particular, for each $A \in \mathcal{D}_{f_1}$, $(t,s) \mapsto \mathfrak{U}_{t,s}\Ab{A}$ is continuous with respect to the norm $\Vert \cdot \Vert_{f_2}$. This   also holds for $(t,s) \mapsto \mathrm{e}^{\mathrm{i}s\mathcal{L}_{H(t)}}$. 
\end{lem}
\begin{proof}
	The proof is inspired by the one of Lemma 4.4 in \cite{moon2019automorphic}. 
	Let $A \in \mathcal{D}_{f_1}$. From Corollary~\ref{lrbcorr}, we have
	\begin{align*}
 \left\Vert (\mathfrak{U}^{\Lambda_N}_{t,s}  - \; \mathfrak{U}_{t,s})\Ab{A}\right\Vert \; \le  & \;\left\Vert\left( \mathfrak{U}_{t,s}^{\Lambda_N}  - \mathfrak{U}_{t,s}\right)\circ \mathbb{E}_{\Lambda_{\left\lfloor \frac{\lambda N}{2 }\right\rfloor}}\Ab{A}\right \Vert\\
 &+ \ \left\Vert  \left(\mathfrak{U}^{\Lambda_N}_{t,s} - \mathfrak{U}_{t,s}\right) \circ\big(1-  \mathbb{E}_{\Lambda_{\left\lfloor \frac{\lambda N}{2 }\right\rfloor}}\big)\Ab{A} \right\Vert \\
 \le & \; \,C \, N^{d+1} \vert t-s\vert \mathrm{e}^{a(v_a\vert t-s \vert - (\lambda N-\left\lfloor \frac{\lambda N}{2 }\right\rfloor)^{\gamma})} \Vert A \Vert + 2 f_1\left(\left\lfloor \lambda N /2\right\rfloor\right) \Vert A \Vert_{f_1} \\
 \le & \;  \left( C \, N^{d+1} \vert t-s\vert \mathrm{e}^{a(v_a\vert t-s \vert - (\lambda N-\left\lfloor \frac{\lambda N}{2 }\right\rfloor)^{\gamma})}  + 2 f_1\left(\left\lfloor \lambda N/2\right\rfloor\right) \right)\Vert A \Vert_{f_1}.
	\end{align*}
Similarly, again applying Corollary~\ref{lrbcorr} with $\lambda = \lambda_2$ for $M \le \left\lfloor \frac{\lambda N}{2}\right\rfloor$, we have
	\begin{align*}
  &\hspace{-15pt}\left\Vert (1 - \mathbb{E}_{\Lambda_M})\circ \left(\mathfrak{U}_{t,s}^{\Lambda_N} - \mathfrak{U}_{t,s}\right)\Ab{A}\right \Vert \\ 
& \; \le  \;  \left\Vert \left(\mathfrak{U}_{t,s}^{\Lambda_N} - \mathfrak{U}_{t,s} \right)\circ \mathbb{E}_{\Lambda_{\left\lfloor \frac{\lambda N}{2 }\right\rfloor}}\Ab{A} \right \Vert \ + \ \left\Vert \mathbb{E}_{\Lambda_M}\circ \left(\mathfrak{U}_{t,s}^{\Lambda_N} -  \mathfrak{U}_{t,s}  \right)
\circ\mathbb{E}_{\Lambda_{\left\lfloor \frac{\lambda N}{2 }\right\rfloor}}\Ab{A} \right \Vert \\
&\qquad+ \, 4\left\Vert \big( \mathbb{E}_{\Lambda_{\left\lfloor \frac{\lambda N}{2 }\right\rfloor}}-1\big)\Ab{A}  \right\Vert  \\
& \; \le  \;  2\, C \, N^{d+1}  \,\vert t-s\vert  \,\mathrm{e}^{a(v_a\vert t-s \vert - (\lambda N-\left\lfloor \frac{\lambda N}{2 }\right\rfloor)^{\gamma})} \Vert A \Vert \;+ \;4 \, f_1\left(\left\lfloor \lambda N /2\right\rfloor\right) \Vert A \Vert_{f_1} \\
& \; \le  \;  \left(2\, C \, N^{d+1} \, \vert t-s\vert  \,\mathrm{e}^{a(v_a\vert t-s \vert - (\lambda N-\left\lfloor \frac{\lambda N}{2 }\right\rfloor)^{\gamma})} \; +\; 4  \,f_1\left(\left\lfloor \lambda N/2\right\rfloor\right) \right)\Vert A \Vert_{f_1}. 
\end{align*}
On the other hand by application of Proposition~\ref{lrb} and Corollary~\ref{lrbcorr} in combination with Lemma~C.2 from \cite{henheikteufel20202}, for $M \ge \left\lfloor \frac{\lambda N}{2}\right\rfloor$, we have
\begin{align*}
 \left\Vert  \left(1 - \mathbb{E}_{\Lambda_M} \right) \circ  \mathfrak{U}_{t,s}^{\Lambda_N} \circ \mathbb{E}_{\Lambda_{\left\lfloor \frac{M}{2 }\right\rfloor}}\Ab{A} \right\Vert  
\, 
&\le\,  C \Vert A \Vert M^d \mathrm{e}^{a(v_a\vert t-s \vert - (M-\left\lfloor \frac{M}{2}\right\rfloor)^\gamma)} \\
 \left\Vert  \left(1 - \mathbb{E}_{\Lambda_M} \right) \circ  \mathfrak{U}_{t,s}  \circ \mathbb{E}_{\Lambda_{\left\lfloor \frac{M}{2 }\right\rfloor}}\Ab{A} \right\Vert  
\,  
&\le\, C \Vert A \Vert M^d \mathrm{e}^{a(v_a\vert t-s \vert - (M-\left\lfloor \frac{M}{2}\right\rfloor)^\gamma)} \,.
\end{align*}
Therefore, for $M \ge \left\lfloor \frac{\lambda N}{2}\right\rfloor$, we have
\begin{align*} 
  &\hspace{-15pt}\left\Vert (1 - \mathbb{E}_{\Lambda_M})\circ \left(\mathfrak{U}_{t,s}^{\Lambda_N} - \mathfrak{U}_{t,s}\right)\Ab{A}\right \Vert \\ 
&\;\le \;2\, C \,\Vert A \Vert \,M^d \,  \mathrm{e}^{a(v_a\vert t-s \vert - (M-\left\lfloor \frac{M}{2}\right\rfloor)^\gamma)}  \;+\; 4 \,f_1\left(\left\lfloor M/2\right\rfloor\right) \Vert A \Vert_{f_1} \\&\;\le \; \left(2 \,C\,  M^d \,  \mathrm{e}^{a(v_a\vert t-s \vert - (M-\left\lfloor \frac{M}{2}\right\rfloor)^\gamma)}  \;+ \;4\, f_1\left(\left\lfloor M/2\right\rfloor\right)\right) \Vert A \Vert_{f_1}\,.
\end{align*}
By collecting the estimates above and using the conditions on $f_1$ and $f_2$, we find that 
\begin{equation*}
\left\Vert( \mathfrak{U}^{\Lambda_N}_{t,s}  -  \mathfrak{U}_{t,s})\Ab{A}\right\Vert_{f_2} \le h_{t,s}(N) \Vert A \Vert_{f_1}, 
\end{equation*}
where $h_{t,s}(N) \to 0$ as $N \to \infty$ uniformly for $s$ and $t$ in compacts.
\end{proof}
\begin{lem} {\rm (Quasi-locality of conjugation with unitaries)} \label{quasilocalautomorphism}~\\
	Let $S \in \mathcal{L}_{I, \zeta, 0 }$ have a   thermodynamic limit    and let $\mathrm{e}^{\mathrm{i}\mathcal{L}_{S(t)}}$ denote the automorphism on $\mathcal{A}$ defined by 
	\[
	\mathrm{e}^{\mathrm{i}\mathcal{L}_{S(t)}}\Ab{A} := \lim_{k\to \infty}\mathrm{e}^{\mathrm{i}\mathcal{L}^{\Lambda_k}_{S(t)}}\Ab{A}
	:= \lim_{k\to \infty}     \mathrm{e}^{-\mathrm{i} S^{\Lambda_k}(t)} \,A  \,   \mathrm{e}^{\mathrm{i} S^{\Lambda_k}(t)} 
	\]
	 for $A\in \mathcal{A}_\mathrm{loc}$. 
	Let $f_1,f_2 :[0,\infty) \to (0,\infty)$ be bounded, non-increasing functions with $\lim_{t\to \infty} f_i(t) = 0$, for $i = 1,2$, such that
	\begin{align*}
	\sup_{N \in \mathbb{N}}  \frac{f_1(N-\lfloor \frac{N}{2} \rfloor)}{f_2(N)}  < \infty\,,\qquad\mbox{and}\qquad 
	\sup_{N \in \mathbb{N}}  \frac{N^d \, \zeta(\left\lfloor \frac{N}{2} \right\rfloor)}{f_2(N)}  < \infty\,.
	\end{align*}
 Then $\mathrm{e}^{\mathrm{i}\mathcal{L}_{S(t)}}: \mathcal{D}_{f_1} \to \mathcal{D}_{f_2}$ is a bounded operator and the sequence $\mathrm{e}^{\mathrm{i}\mathcal{L}^{\Lambda_N}_{S(t)}}: \mathcal{D}_{f_1} \to \mathcal{D}_{f_2}$ of operators is uniformly bounded, both uniformly in $t \in I$. More precisely, there exists a constant $C_{f_1,f_2}$ such that
	\begin{equation*}
	\sup_{t \in I} \left\Vert \mathrm{e}^{\mathrm{i}\mathcal{L}_{S(t)}}\Ab{A} \right\Vert_{f_2} \le C_{f_1,f_2} \Vert A\Vert_{f_1} \quad \text{and} \quad \sup_{N \in \mathbb{N}} \sup_{t \in I}\left\Vert \mathrm{e}^{\mathrm{i}\mathcal{L}^{\Lambda_N}_{S(t)}}\Ab{A} \right\Vert_{f_2} \le C_{f_1,f_2} \Vert A\Vert_{f_1}
	\end{equation*}
	for all $A \in \mathcal{D}_{f_1}$. Beside the indicated dependence on $f_1$ and $f_2$, the constant $C_{f_1,f_2}$ only depends on $\zeta$ and the norm $\Vert S \Vert_{I, \zeta, 0 }$. 
\end{lem}
\begin{proof}
As a conjugation with unitaries can be viewed as a time-evolution at a frozen time, we have by application of Corollary~\ref{lrbcorr} and Lemma~C.2 from \cite{henheikteufel20202} (analogous to the proof of Lemma 4.5 in \cite{moon2019automorphic}) that 
\begin{align*}
 \left\Vert  \left( 1 -  \mathbb{E}_{\Lambda_N}\right) \circ \mathrm{e}^{\mathrm{i}\mathcal{L}_{S(t)}}\Ab{A} \right\Vert  \;\le\; &\left\Vert \left(1 - \mathbb{E}_{\Lambda_N}\right)\circ \mathrm{e}^{\mathrm{i}\mathcal{L}_{S(t)}}\circ \mathbb{E}_{\Lambda_k}\Ab{A} \right\Vert + 2 \left\Vert A - \mathbb{E}_{\Lambda_k}\Ab{A} \right\Vert \\
\le \; &C \Vert A \Vert k^d \zeta(N-k)+ 2 \Vert A\Vert_{f_1} f_1(k) 
\end{align*}
for $N,k \in \mathbb{N}$ with $k<N$, uniformly in $t \in I$. 
Using this bound with $k = N- \left\lfloor \frac{N}{2}\right\rfloor$, yields the claim. The boundedness of   $\mathrm{e}^{\mathrm{i}\mathcal{L}^{\Lambda_N}_{S(t)}}$ can be proven in the same way. 
\end{proof}
\subsection{Derivations}
\begin{prop} {\rm (Commutator bounds and convergence of derivations)} \label{tdlofderivations} ~\\
	Let $H_0 \in \mathcal{L}_{I, \zeta, 0  }$ and $v \in \mathcal{V}_I$ both have a   thermodynamic limit. Define the induced family of local derivations as
	\begin{equation*}
	\mathcal{L}_t^{\Lambda_k} : \mathcal{A}_{\Lambda_k} \to \mathcal{A}_{\Lambda_k}\,,\quad A \mapsto \mathcal{L}_t^{\Lambda_k} \Ab{A} := \left[H_0^{\Lambda_k}(t) + V_v^{\Lambda_k}(t), A\right]\,.
	\end{equation*}
	{Then for all $k \in \mathbb{N}$ and for all $X, Y \subset\Lambda_k$, $A \in \mathcal{A}_X$ and $B \in \mathcal{A}_Y$ it holds that }
	\begin{align*}
	\sup_{t \in I}\Vert \mathcal{L}^{\Lambda_k}_t\Ab{A} \Vert &\;\le\; C \, \Vert A \Vert \, \mathrm{diam}(X)^{d+1} , \\
	\sup_{t \in I}\Vert [\mathcal{L}^{\Lambda_k}_t\Ab{A}, B] \Vert &\;\le\; 2\, C \, \Vert A \Vert \,  \Vert B \Vert \, \mathrm{diam}(X)^{d+1}  \zeta(\mathrm{dist}^{\Lambda_k}(X,Y)).
	\end{align*}
	 If $H_0$ and $v$ have a rapid thermodynamic limit with exponent $\gamma \in (0,1)$, then 
		there exist  $\lambda_1 >0$, $\lambda_2 \in (0,1)$, and $C<\infty$,  such that for all $l,k \in \mathbb{N}$ with $l \ge k$, $X \subset \Lambda\left(k \right)$ and $A \in \mathcal{A}_X$
	\begin{align*}
 \sup_{t \in I} \Vert \mathcal{L}_t^{\Lambda_l}\Ab{A}  - \mathcal{L}_t^{\Lambda_k}\Ab{A} \Vert \le \, C \,  \Vert A\Vert \, \mathrm{diam}(X)^{d+1} \ \zeta_{\gamma}\left(\mathrm{dist}^{\Lambda_l}(X, \Lambda_l\setminus  
 \Lambda_{\max\{ \left\lceil k - \lambda_1 k^{\gamma}\right\rceil, \left\lceil\lambda_2 \cdot k\right\rceil\}} )
 \right) \,.
\end{align*}
In every case above, the   constant $C$  only depends on $\zeta$,   $\Vert \Phi_{H_0}\Vert_{I,\zeta, 0 }$, and   $C_v$.
\end{prop}
\begin{proof}
	Involving Lemma~\ref{lipschitzcomm}, the first and the second estimate are clear from Example~5.4 in \cite{nachtergaele2019quasi}. The first alternative of the third estimate is proven analogous to Theorem~3.8~(ii) in \cite{nachtergaele2019quasi} and choosing $\Lambda_M = \Lambda_{\lceil k - \lambda_1 k^{\gamma}\rceil}$ with $\lambda_1$ from Definition~\ref{cauchydefinition}  resp.~the alternative characterisation from Lemma \ref{lem:equivalentchar}. The second alternative can easily be concluded from the first. This latter form is used in the proof of the lemmata below.
\end{proof}
As a consequence, $(\mathcal{L}_t^{\Lambda_k}(A))_{k \in \mathbb{N}}$ is a Cauchy sequence for all $A \in \mathcal{A}_{\mathrm{loc}}$ uniformly for $t \in I$. The limiting derivation $(\mathcal{L}_t, D(\mathcal{L}_t))$ is closable (see Proposition~3.2.22 and Proposition~3.1.15 in \cite{bratteli2012operator}) and its closure (denoted by the same symbol) is called the bulk derivation generated by $H = H_0+V_v$ at $t \in I$.
 
  Since the finite volume metrics $d^{\Lambda_k}(\cdot, \cdot)$ are compatible in the bulk, we obtain the corresponding statements also for the bulk derivation.  
\begin{cor} {\rm (Infinite volume derivation)} \label{tdlofderivationscorr} ~\\
Under the conditions of Proposition~\ref{tdlofderivations}, there exists   and $C<\infty$,  such that for all $X,Y \in \mathcal{P}_0(\Gamma)$, $A \in \mathcal{A}_X$ and $B \in \mathcal{A}_Y$
\begin{align*}
	\sup_{t \in I}\Vert \mathcal{L}_t\Ab{A}  \Vert &\;\le\; C \, \Vert A \Vert \, \mathrm{diam}(X)^{d+1} ,  \\
\sup_{t \in I}\Vert [\mathcal{L}_t\Ab{A} , B] \Vert &\;\le\; 2\, C \, \Vert A \Vert \,  \Vert B \Vert \, \mathrm{diam}(X)^{d+1}    \zeta(\mathrm{dist}(X,Y)) \,.
\end{align*}
If $H_0$ and $v$ have a rapid thermodynamic limit, there exist
$\lambda_1 >0$, $\lambda_2 \in (0,1)$ such that
\begin{align*}
\sup_{s \in I} \Vert (\mathcal{L}_t  - \mathcal{L}_t^{\Lambda_k})\Ab{A}  \Vert\;\le\;  C \,  \Vert A\Vert \, \mathrm{diam}(X)^{d+1}  \, \zeta_{\gamma}\left(\mathrm{dist}(X, \Gamma\setminus  
\Lambda_{\max\{ \left\lceil k - \lambda_1 k^{\gamma}\right\rceil, \left\lceil\lambda_2 \cdot k\right\rceil\}} )
\right) \,.
\end{align*}
  In every case above, the   constant $C$   depends only on $\zeta$,   $\Vert \Phi_{H_0}\Vert_{I,\zeta, 0}$, and   $C_v$. 
\end{cor}

 The following two lemmata generalise  Lemma 4.12 and Lemma 4.13 from \cite{moon2019automorphic}. 
Here we use  that whenever $\Psi_{H_0} \in \mathcal{B}^\circ_{I,\exp(-a\cdot),0}$ satisfies (I2), then
 \begin{equation} \label{eq:b}
 \lim_{\delta \to 0}  b(\delta) := \lim_{\delta \to 0}   \sup_{\substack{t,t_0 \in I \\ 0< \vert t-t_0\vert < \delta}} \left\Vert  \frac{\Psi_{H_0}(t)-\Psi_{H_0}(t_0)}{t-t_0} -\dot{\Psi}_{H_0}(t_0)\right\Vert^\circ_{\exp(-a\cdot), 0 }=0\,.
 \end{equation}

\begin{lem} {\rm (Quasi-locality of derivations I)} \label{domainofderivlemma} ~\\Let $H_0 \in \mathcal{L}_{I,\zeta, 0 }$ and $v \in \mathcal{V}_I$ both have a   thermodynamic limit, set $H=H_0+V_v$ and 
	let $\mathcal{L}_s: D(\mathcal{L}_s) \to \mathcal{A}$ be the bulk derivation generated by $H$ at $s \in I$. 
	Let $f_1,f_2: [0,\infty) \to (0,\infty)$ be bounded, non-increasing functions with $\lim\limits_{t \to \infty}f_i(t) =0$, for $i= 1,2$, such that 
	\begin{equation*}
	\sum_{k=1}^{\infty} (k+1)^{d+1} f_1(k) < \infty
	\end{equation*}
	and  
	\begin{equation*}
	\lim\limits_{N \to \infty} \frac{\sum_{k=\lfloor \frac{N}{2}\rfloor}^{\infty}(k+1)^{d+1} f_1(k)}{f_2(N)} = \lim\limits_{N \to \infty} \frac{N^{d+1} \zeta(N-\lfloor\frac{N}{2} \rfloor)}{f_2(N)} =0. 
	\end{equation*}
	  Then $\mathcal{D}_{f_1} \subset D(\mathcal{L}_s)$ and  $\mathcal{L}_t: \mathcal{D}_{f_1} \to  \mathcal{D}_{f_2}$ is a bounded operator and the sequence $\mathcal{L}^{\Lambda_N}_t: \mathcal{D}_{f_1} \to  \mathcal{D}_{f_2}$ of operators   is uniformly bounded, both uniformly in $t \in I$.
	 More precisely, there is a constant $C_{f_1,f_2} > 0 $ such that
	\begin{equation*}
	\sup_{t \in I} \Vert \mathcal{L}_t\Ab{A} \Vert_{f_2} \ \le \ C_{f_1,f_2}\Vert A \Vert_{f_1} \ \ \text{and} \ \ \sup_{N \in \mathbb{N}} \sup_{t \in I}\Vert \mathcal{L}_t^{\Lambda_N}\Ab{A}  \Vert_{f_2} \ \le \ C_{f_1,f_2} \Vert A \Vert_{f_1}
	\end{equation*}
	for all $A \in \mathcal{D}_{f_1}$. Beside the indicated dependence on $f_1$ and $f_2$, the constant $C_{f_1,f_2}$ only depends on $\zeta$,     ${\Vert \Phi_{H_0} \Vert_{I,\zeta, 0 }}$, and   $C_v$. 
	
	In particular, let $H_0$ be given by an infinite volume interaction $\Psi_{H_0} \in \mathcal{B}_{I, \exp(-a \cdot), 0}^\circ$ that satisfies (I2). Then, with $b(\delta)$ defined in \eqref{eq:b}, it follows that
	\begin{align*}
	\sup_{N \in \mathbb{N}} \sup_{\substack{t,t_0 \in I \\ 0< \vert t-t_0\vert < \delta}} \left\Vert  \mathcal{L}^{\Lambda_N}_{\frac{H_0(t) - H_0(t_0)}{t-t_0} - \dot{H}_0(t_0)}\Ab{A}  \right\Vert _{f_2} \;\le\; b(\delta)\, C_{f_1,f_2} \,\Vert A \Vert_{f_1}
	\end{align*}
	and
		\begin{align*}
	\sup_{\substack{t,t_0 \in I \\ 0< \vert t-t_0\vert < \delta}} \left\Vert  \mathcal{L}_{\frac{H_0(t) - H_0(t_0)}{t-t_0} - \dot{H}_0(t_0)}\Ab{A}  \right\Vert _{f_2}\; \le\; b(\delta) \,C_{f_1,f_2} \,\Vert A \Vert_{f_1}.
	\end{align*} 
\end{lem}
\begin{proof}
	The proof of this Lemma is analogous to the one of Lemma 4.12 in \cite{moon2019automorphic}.  By application of Corollary~\ref{tdlofderivationscorr}, for any $A \in \mathcal{D}_{f_1}$ and $N,M \in \mathbb{N}$ with $N>M$, we have
	\begin{align}
	 \Vert \mathcal{L}_t\circ (\mathbb{E}_{\Lambda_N}  - \mathbb{E}_{\Lambda_M})\Ab{A}  \Vert &\;=\; \left\Vert \sum_{k=M}^{N-1} \mathcal{L}_t\circ (\mathbb{E}_{\Lambda_{k+1}} - \mathbb{E}_{\Lambda_k} )\Ab{A} \right\Vert \nonumber\\ &\;\le\;  C \sum_{k=M}^{N-1} (k+1)^{d+1} \Vert (\mathbb{E}_{\Lambda_{k+1}} - \mathbb{E}_{\Lambda_k})\Ab{A} \Vert \nonumber\\  &\;\le\; 2\,C \left(\sum_{k=M}^{N-1} (k+1)^{d+1} f_1(k) \right) \Vert A \Vert_{f_1}. \label{derivationcauchy}
	\end{align}
	This implies, that $(\mathcal{L}_t\circ \mathbb{E}_{\Lambda_N}\Ab{A} )_{N \in \mathbb{N}}$ with $A \in \mathcal{D}_{f_1}$ is a Cauchy sequence in $\mathcal{A}$, hence there exists a limit. Moreover, $\mathbb{E}_{\Lambda_N}\Ab{A} $ converges to $A$ in $\Vert \cdot \Vert $. Since the derivation is closed, $A \in \mathcal{D}_{f_1}$ belongs to the domain $D(\mathcal{L}_t)$ of $\mathcal{L}_t$ and 
	\begin{equation*}
	\mathcal{L}_t\Ab{A}  = \lim\limits_{N \to \infty} \mathcal{L}_t\circ \mathbb{E}_{\Lambda_N}\Ab{A} .
	\end{equation*}
	Hence, $\mathcal{D}_{f_1} \subset D(\mathcal{L}_t)$. Similarly to the estimate in \eqref{derivationcauchy}, one obtains
	\begin{equation*}
	\Vert \mathcal{L}_t\Ab{A}  \Vert \le C\left(2  \sum_{k=1}^{\infty} (k+1)^{d+1} f_1(k)  + 1\right) \Vert A \Vert_{f_1}
	\end{equation*}
	for any $A \in \mathcal{D}_f$ by considering a limit. 
	Now, we are left to estimate
	\begin{align*}
	 \Vert (1- &\mathbb{E}_{\Lambda_N})\circ \mathcal{L}_t\Ab{A}  \Vert = \lim\limits_{M \to \infty}\Vert  (1- \mathbb{E}_{\Lambda_N})\circ \mathcal{L}_s\circ \mathbb{E}_{\Lambda_M}\Ab{A} \Vert \\
	= \ &\lim_{M \to \infty} \left\Vert \mathcal{L}_t\circ (\mathbb{E}_{\Lambda_M}  - \mathbb{E}_{\Lambda_{\lfloor \frac{N}{2}\rfloor}}  + \mathbb{E}_{\Lambda_{\lfloor \frac{N}{2}\rfloor}} )\Ab{A} -  \mathbb{E}_{\Lambda_N}\circ \mathcal{L}_t\circ (\mathbb{E}_{\Lambda_M} - \mathbb{E}_{\Lambda_{\lfloor \frac{N}{2}\rfloor}} + \mathbb{E}_{\Lambda_{\lfloor \frac{N}{2}\rfloor}})\Ab{A} \right\Vert \\
	\le \  &2 \lim\limits_{M \to \infty}\Vert \mathcal{L}_t\circ (\mathbb{E}_{\Lambda_M}  - \mathbb{E}_{\Lambda_{\lfloor \frac{N}{2}\rfloor}} ) \Ab{A} \Vert  + \Vert  (1 - \mathbb{E}_{\Lambda_N})\circ \mathcal{L}_s\circ \mathbb{E}_{\Lambda_{\lfloor \frac{N}{2}\rfloor}}\Ab{A}  \Vert \\ 
	\le \ &4\,C\Big(\sum_{k=\lfloor \frac{N}{2} \rfloor}^{\infty} (k+1)^{d+1} f_1(k)  + N^{d+1} \zeta(N- \left\lfloor\tfrac{N}{2}\right\rfloor) \Big) \Vert A \Vert_{f_1}. 
	\end{align*}
	In the last line we used Corollary~\ref{tdlofderivationscorr} together with Lemma~C.2 from \cite{henheikteufel20202}. Hence, we have proven the claim. The boundedness of the sequence $\mathcal{L}_t^{\Lambda_N}$ can be proven in the same way.  
\end{proof}
\begin{lem} {\rm (Quasi-locality of derivations II)} \label{convofderivlemma}~\\
	Let $H_0 \in \mathcal{L}_{I,\zeta, 0 }$ and $v \in \mathcal{V}_I$ both have a rapid thermodynamic limit with exponent $\gamma \in (0,1)$, set $H=H_0+V_v$ and 
	let $\mathcal{L}_t: D(\mathcal{L}_t) \to \mathcal{A}$ be the derivation generated by $H$ at $t \in I$. 
	Let $f_1,f_2: [0,\infty) \to (0,\infty)$ be bounded, non-increasing functions  with $\lim\limits_{s\to \infty}f_i(s) =0$, for $i=1,2$, such that 
	\[
	\sum_{k=1}^{\infty} (k+1)^{d+1} \sqrt{f_1}(k) < \infty
	\]
and
	\begin{equation*}
	\lim\limits_{N \to \infty} \frac{\sum_{k=\lfloor \frac{N}{2}\rfloor}^{\infty}(k+1)^{d+1} \sqrt{f_1}(k)}{f_2^2(N)} = \lim\limits_{N \to \infty} \frac{N^{d+1} \zeta_{\gamma}(N - \lfloor \frac{N}{2} \rfloor)}{f_2^2(N)} =0. 
	\end{equation*}
 Then $\mathcal{D}_{f_1} \subset D(\mathcal{L}_t)$, $\mathcal{L}_{t} : \mathcal{D}_{f_1} \to \mathcal{D}_{f_2}$ and $\mathcal{L}^{\Lambda_N}_{t} : \mathcal{D}_{f_1} \to \mathcal{D}_{f_2}$ are bounded operators, and 
	$\mathcal{L}_{t}^{\Lambda_N} \xrightarrow{N \to \infty} \mathcal{L}_{t}$
	in operator norm uniformly for all $t \in I$.
	In particular, for each $A \in \mathcal{D}_{f_1}$, $t \mapsto \mathcal{L}_t(A)$ is continuous with respect to the norm $\Vert \cdot \Vert_{f_2} $.  
\end{lem}
\begin{proof}
	The proof of this Lemma is analogous to the one of Lemma 4.13 in \cite{moon2019automorphic}.
	From Corollary~\ref{tdlofderivationscorr}, combined with Lemma~\ref{domainofderivlemma}, we have 
	 for $A \in \mathcal{D}_{f_1}$ that 
	\begin{align*}
	\left\Vert \left( \mathcal{L}_t\right. \right. - &\left. \left.  \mathcal{L}_s^{\Lambda_N}\right)\Ab{A}  \right\Vert \\ 
	\le \ &\left\Vert \left( \mathcal{L}_t-\mathcal{L}_t^{\Lambda_N}\right) \circ \left(1-\mathbb{E}_{\Lambda_{\lfloor \frac{\lambda N}{2} \rfloor}} \right) \Ab{A} \right\Vert \;+\; \left\Vert \left( \mathcal{L}_t-\mathcal{L}_t^{\Lambda_N}\right)\circ  \mathbb{E}_{\Lambda_{\lfloor \frac{\lambda N}{2} \rfloor}}\Ab{A}  \right\Vert \\ 
	\le \  &2\, C_{\sqrt{f_1},f_2^2} \Vert \left(1-\mathbb{E}_{\Lambda_{\lfloor \frac{\lambda N}{2} \rfloor}}\right)\Ab{A}  \Vert_{\sqrt{f_1}} \;+\; C \Vert A\Vert  N^{d+1}   \, \zeta_{\gamma}(\lambda N -\lfloor \lambda N/2\rfloor ) \\
	\le \ &\left(2\, C_{\sqrt{f_1},f_2^2} \left(f_1(\lfloor \lambda N/2\rfloor) \;+\; \sqrt{f_1}(\lfloor \lambda N/2\rfloor)\right) + C \, N^{d+1}   \, \zeta_{\gamma}(\lambda N - \lfloor \lambda N/2\rfloor)\right) \Vert A\Vert_{f_1},
	\end{align*}
  which vanishes as $N \to \infty$. Therefore,  
	\begin{equation*}
	\lim\limits_{N \to \infty} \sup_{t \in I} \left\Vert   \mathcal{L}_t - \mathcal{L}_t^{\Lambda_N}   \right\Vert_{\mathcal{L}(\mathcal{D}_{f_1},\mathcal{A} )} = 0 \,. 
	\end{equation*} 
	Furthermore, for $A \in \mathcal{D}_{f_1}$, we have by application of Lemma~\ref{domainofderivlemma} that 
	\begin{align*}
	&\hspace{-5pt}\frac{\left\Vert   \left( 1- \mathbb{E}_{\Lambda_M}\right)\circ  \left( \mathcal{L}_t - \mathcal{L}_t^{\Lambda_N}\right) \Ab{A}  \right\Vert }{f_2(M)} \\ 
	& \quad\le  \begin{cases}
	2 \, C_{\sqrt{f_1},f_2^2}  \, f_2(\lfloor \lambda N/2 \rfloor) \Vert A \Vert_{\sqrt{f_1}}, \ \text{for} \  M> \lfloor \frac{\lambda N}{2} \rfloor \\
	\, \\
	\frac{4 C_{\sqrt{f_1},f_2^2} \left(f_1(\lfloor \lambda N/2\rfloor) + \sqrt{f_1}(\lfloor \lambda N/2\rfloor)\right) + 2C \, N^{d+1}   \, \zeta_{\gamma}(\lambda N - \lfloor \lambda N/2 \rfloor)}{f_2(\lfloor \lambda N/2 \rfloor)} \Vert A \Vert_{f_1}, \ \text{for} \ M \le \lfloor \frac{\lambda N}{2} \rfloor,
	\end{cases}
	\end{align*}
	which implies the claim.  
\end{proof}

\subsection{Inverse Liouvillian}\label{inverseL}
Finally, we show that also  $\mathcal{I}_t$  defined in Appendix~\ref{invliouappendix} preserves quasi-locality of observables. Similar statements hold for $\mathcal{J}_t$. 
  The following lemma generalises Lemma~4.6 from \cite{moon2019automorphic}.  
\begin{lem} {\rm (Quasi-locality of the inverse Liouvillian)}  \label{invliouconvlemma} ~\\
	Let $H_0 \in \mathcal{L}_{I, \exp(-a \cdot),0 }$ have a   thermodynamic limit.
 For $f_1, f_2$ as in Lemma~\ref{quasilocaltimeevolution1}, $\mathcal{I}_t : \mathcal{D}_{f_1} \to \mathcal{D}_{f_2}$ is a bounded operator and the sequence $\mathcal{I}^{\Lambda_N}_t : \mathcal{D}_{f_1} \to \mathcal{D}_{f_2}$ is uniformly bounded, both uniformly in $t \in I$. More precisely, there exists a positive constant $C_{f_1,f_2}$ such that 	for all $ A \in \mathcal{D}_{f_1} $
	\begin{align*}\sup_{t \in I}\Vert\mathcal{I}_t\Ab{A}\Vert_{f_2}\le C_{f_1,f_2}\Vert A\Vert_{f_1}\quad\text{and}\quad\sup_{N\in\mathbb{N}}\sup_{t\in I}\Vert\mathcal{I}_t^{\Lambda_N}\Ab{A}\Vert_{f_2}\le C_{f_1,f_2}\Vert A\Vert_{f_1}\,.\end{align*}
  	Beside the indicated dependence on $f_1$ and $f_2$, the constant $C_{f_1,f_2}$ only depends on $a$ and   ${\Vert \Phi_0 \Vert_{I,\exp(-a \, \cdot ), 0 }}$.
	
	If $H_0$ has a rapid thermodynamic limit with exponent $\gamma \in (0,1)$ and if
  the functions $f_1$and  $f_2$  satisfy, in addition,  the requirements of Lemma~\ref{quasilocaltimeevolution2}, then ${\mathcal{I}_t^{\Lambda_N} \xrightarrow{N \to \infty} \mathcal{I}_t}$ in operator norm uniformly in $t \in I$.  
	
	The above statements also hold for the derivatives $\frac{\D^k}{\D t^k} \mathcal{I}_t$ resp.\ $\frac{\D^k}{\D t^k} \mathcal{I}^{\Lambda_N}_t$ after suitably replacing the conditions on $f_1$ and $f_2$ from Lemma~\ref{quasilocaltimeevolution1} and Lemma~\ref{quasilocaltimeevolution2}.  
\end{lem}
\begin{proof}
	The first part of this Lemma is a simple consequence of Lemma~\ref{quasilocaltimeevolution1} using the properties of the function $b_{f_1,f_2}$. 
	The second part follows from Lemma~\ref{quasilocaltimeevolution2} and Lemma~\ref{quasilocaltimeevolution1} in combination with the dominated convergence theorem. 
	
 The statements on the derivatives follow with the aid of Lemma \ref{quasilocaltimeevolution1}, Lemma \ref{quasilocaltimeevolution2}, Lemma~\ref{domainofderivlemma}, and Lemma~\ref{convofderivlemma} by noticing that the first derivative is given by
		\begin{equation*}
	\frac{\mathrm{d}}{\mathrm{d}t} \mathcal{I}_t\Ab{A}  = \mathrm{i} \int_{\mathbb{R}} \mathrm{d}s \ w(s) \, s^2 \int_{0}^{1} \D u \, u \int_{0}^{1} \D r \ \mathrm{e}^{\mathrm{i}(1-r)us \mathcal{L}_{H_0(t)}} \circ \mathcal{L}_{\dot{H}(t)} \circ \mathrm{e}^{\mathrm{i}rus \mathcal{L}_{H_0(t)}}\Ab{A} 
	\end{equation*}
	and the higher derivatives have a similar form. More precisely, for this first derivative, sufficient conditions on $f_1$ and $f_2$ for the boundedness of $\frac{\mathrm{d}}{\mathrm{d}t} \mathcal{I}_t : \mathcal{D}_{f_1} \to \mathcal{D}_{f_2}$ are given by
		\[
	\int_{0}^\infty \mathrm{d}s \ w_g(s) \frac{ (2 s)^2 }{f_2(4v_a s) \cdot \sqrt{f_1}(4v_as)} < \infty\,, 
	\]
\[
	\lim\limits_{N \to \infty} \frac{\sum_{k=\lfloor \frac{N}{2}\rfloor}^{\infty}(k+1)^{d+1} \sqrt{f_1}(k)}{f_2^2(N)} = \lim\limits_{N \to \infty} \frac{N^{d+1} \E^{-a\frac{N}{4}}}{f_2^2(N)} =0. 
\]
This follows by application of Lemma \ref{quasilocaltimeevolution1} with $(f_1,\sqrt{f_1})$, Lemma \ref{domainofderivlemma} with $(\sqrt{f_1}, f_2^2)$, and again Lemma \ref{quasilocaltimeevolution1} with $(f_2^2,f_2)$. For higher derivatives, the powers of $f_1$ in the conditions above get (polynomially) decreased, whereas the powers of $f_2$ and $ t $ get (polynomially) increased. Using Lemma \ref{convofderivlemma}, one can easily establish sufficient conditions on $f_1$ and $f_2$ for the convergence ${\frac{\mathrm{d}^k}{\mathrm{d}t^k} \mathcal{I}_t^{\Lambda_N} \xrightarrow{N \to \infty} \frac{\mathrm{d}^k}{\mathrm{d}t^k} \mathcal{I}_t}$ in operator norm involving even higher resp.~lower powers of $f_2$, $s$, and $f_1$.  
\end{proof}

\subsection{Sequences of suitable weight functions}

\begin{lem}\label{SeqLem}
  Let $v_a>0$ (Lieb-Robinson velocity),  $\zeta \in \mathcal{S}$, $R\in (1,\infty)$ and define the set $\mathcal{S}_\zeta^R$ of `admissible' functions as
\begin{equation*}
\mathcal{S}_\zeta^R:=\left\{ f \in \mathcal{S}\,\Big|\, \int_{0}^\infty \hspace{-1mm} \mathrm{d}t \; w_g(s)
\left(\frac{ s }{f(4v_as) }\right)^R < \infty, \   \lim\limits_{N \to \infty} \frac{N^{d+1} \,\zeta(N)}{ f(N) ^R} =0 \right\} \,.
\end{equation*}
Then there exists an infinite sequence $\left(f_j\right)_{j \in \mathbb{N}} \subset \mathcal{S}_\zeta^{R}$ satisfying 
\begin{equation*}
\lim\limits_{N \to \infty} \frac{\sum_{k=\lfloor \frac{N}{2}\rfloor}^{\infty}(k+1)^{d+1} \, f_j(k) ^\alpha}{ f_{j+1}(N)^\beta} = 0 \quad \forall \alpha, \beta \in \left[\tfrac{1}{R},R\right] \quad \mathrm{and} \quad \forall j \in \mathbb{N}. 
\end{equation*}
\end{lem}
\begin{proof}
According to Lemma~A.1 from \cite{monaco2017adiabatic}, $\mathcal{S}_\zeta^R $ is non-empty. So, pick any $f \in \mathcal{S}_\zeta^R$ and set
\begin{equation*}
f_j(s) :=  f(s) ^{\frac{1}{ (5R^2)^j } }
\end{equation*}
for any $j \in \mathbb{N}$. Obviously, $f_j \in \mathcal{S}_\zeta^R$ and we estimate
\begin{align*}
 \frac{\sum_{k=\lfloor \frac{N}{2}\rfloor}^{\infty}(k+1)^{d+1}   f_j(k) ^\alpha}{ f_{j+1}(N) ^\beta}  &
\;\le\; \left( \sum_{k=\lfloor \frac{N}{2}\rfloor}^{\infty}(k+1)^{d+1}  f_j(k) ^{\frac{1}{2R}} \right) \frac{ f_{j}\left(\lfloor N/2 \rfloor \right) ^{\frac{1}{2R}}}{ f_{j+1}(N) ^R} \\
&\;\le\; C_j \, \frac{ f_{j}\left(N/2 \right) ^{\frac{1}{2R}}}{ f_{j+1}(N/2) ^{2R}}\; =\; C_j \, f(N/2)^{\frac{1}{(5R^2)^j} { \frac{1}{10R}}}\xrightarrow{N \to \infty} 0
\end{align*}
uniformly in $\alpha, \beta \in \left[\tfrac{1}{R},R\right]$ and for any $j \in \mathbb{N}$. In the second inequality, we used that  $f_{j+1}$ is logarithmically superadditive  and $ f_j(1/2) ^{\frac{1}{2R}} > 0$. 
\end{proof}

\section{Operations preserving the rapid thermodynamic limit}\label{app:C}
In this appendix, we verify that operator families that are obtained by taking commutators and inverse Liouvillians (see Appendix~\ref{invliouappendix}) of operator families and Lipschitz potentials having a rapid thermodynamic limit, again have a rapid thermodynamic limit with the same exponent. The following statements are  slight modifications of technical lemmata from~\cite{henheikteufel20202}. 
\begin{lem}{\rm (Alternative characterisation of having a rapid thermodynamic limit)} \label{lem:equivalentchar} \\ 
	The time-dependent interaction $\Phi \in \mathcal{B}_{I,\zeta,n}$ has a rapid thermodynamic limit if and only if it satisfies the following Cauchy property:
	\begin{align*}
\forall i \in \mathbb{N}_0 \ \ &\exists \lambda,  C>0 \ \ \forall M\in \N\ \ \forall \,   k,l \, \ge \, M + \lambda M^{\gamma} : \\ &\sup_{t \in I}\left\Vert  \frac{\D^i}{\D t^i}\left( \Phi^{\Lambda_l}- \Phi^{ \Lambda_k}\right) (t) \, \right\Vert_{\zeta^{(i)}, n,\Lambda_M}  \le C \, \zeta^{(i)}(M^{\gamma}) =: C \zeta^{(i)}_{\gamma}(M). \nonumber
\end{align*}
\end{lem}
\begin{proof} 
The proof is completely analogous to the one of Lemma D.1 in \cite{henheikteufel20202}. 
\end{proof}
\begin{lem} {\rm (Commutator of SLT-operators)} \label{cauchy1}~\\
	Let $A, B \in \mathcal{L}_{I,\zeta, n+d }$ have a rapid thermodynamic limit with exponent $\gamma \in (0,1)$. Then the commutator $[A,B] \in \mathcal{L}_{I,\zeta, n }$ also has a rapid thermodynamic limit with exponent $\gamma $. In particular, if $A,B \in \mathcal{L}_{I,\mathcal{S}, \infty}$ both have a thermodynamic limit with exponent $\gamma $, then $[A,B] \in \mathcal{L}_{I,\mathcal{S}, \infty }$ also has a thermodynamic limit with exponent $\gamma $. 
\end{lem}
\begin{proof} 
	Using Lemma \ref{lem:equivalentchar}, the proof is completely analogous to the one of Lemma D.2 in \cite{henheikteufel20202}. 
\end{proof}
\begin{lem} {\rm (Commutator with Lipschitz potential)}  \label{cauchy2} ~\\
	Let $A \in \mathcal{L}_{I,\zeta, n+d+1 }$ and $v \in \mathcal{V}_I$ both have a rapid thermodynamic limit with exponent $\gamma \in (0,1)$. Then the commutator $[A,V_v] \in \mathcal{L}_{I,\zeta, n}$ also has a rapid thermodynamic limit with exponent $\gamma $. In particular, if $A \in \mathcal{L}_{I,\mathcal{S}, \infty }$ has a rapid thermodynamic limit with exponent $\gamma $, then $[A,V_v] \in \mathcal{L}_{I,\mathcal{S}, \infty }$ also has a rapid thermodynamic limit with exponent~$\gamma $. 
\end{lem}
\begin{proof} 
	Using Lemma \ref{lem:equivalentchar}, the proof is completely analogous to the one of Lemma D.3 in \cite{henheikteufel20202}. 
\end{proof}
\begin{lem} {\rm (Inverse Liouvillian)}  \label{cauchy3} ~\\
  Let $H \in \mathcal{L}_{I,\exp(-a \, \cdot), \infty}$ and $B$ either an SLT-operator in $\mathcal{L}_{I,\mathcal{S}, \infty }$ or a Lipschitz potential. Assume that $H$ and  $[H,B] \in \mathcal{L}_{I,\mathcal{S}, \infty }$ both have a rapid thermodynamic limit with exponent $\gamma \in (0,1)$.   Then $\mathcal{I}_{H}(B) \in \mathcal{L}_{I,\mathcal{S}, \infty }$ also has a rapid thermodynamic limit with exponent $\gamma$.
\end{lem}
\begin{proof}
	Using Lemma \ref{lem:equivalentchar}, the proof is analogous to the one of Lemma D.4 in \cite{henheikteufel20202}. But, there is one crucial point to take care of: One arrives at
	\begin{align*}
	&\hspace{-10pt}\left\Vert\left( \Delta_m^{\Lambda_l}- \Delta_m^{\Lambda_k}\right)\Ab{\Phi_B^{\Lambda_k}(Y)}\right\Vert \;\le \;\sqrt{\text{I} \cdot \text{II}} \\ 
	&\; \le \;  4 \ \Lambda_M\text{-diam}(Y)^d \ \big\Vert \Phi_B^{\Lambda_k}(Y) \big\Vert \ \sqrt{\hat{\zeta}(m)} \ \bigg[C(T)+\tilde{C}(T) \ \cdot \bigg( \Vert \Phi_H\Vert_{I,\exp(-a \, \cdot),0 } \times \\ & \quad \times  \bigg(\sup_{y\in \Lambda_M}\sum_{x \in \Lambda_l \setminus\Lambda_{M'}}F_{\exp(-a \, \cdot)}\left(d^{\Lambda_l}(x,y)\right) + \sup_{y\in \Lambda_M}\sum_{x \in \Lambda_k \setminus\Lambda_{M'}}F_{\exp(-a \, \cdot)}\left(d^{\Lambda_k}(x,y)\right) \ \bigg) \\ &  \quad  + \  \Vert F_{\exp(-a \, \cdot)} \Vert \  \left\Vert \Phi_H^{\Lambda_l}-\Phi_H^{\Lambda_k} \right\Vert_{\exp(-a \, \cdot),0,\Lambda_{M'}} \bigg)\bigg]^{\frac{1}{2}}, 
	\end{align*}
	where we choose $M' = M + M^{\gamma}$ (i.e. $\lambda \to \lambda + 1$ in Definition~\ref{cauchydefinition} resp.~Lemma \ref{lem:equivalentchar}) and $T = \frac{aM^{\gamma}}{8 \Vert \Phi_H\Vert }$, such that the term in the square brackets decays faster than any polynomial as $M \to \infty$. So, after a possible adjustment of $\zeta \in \mathcal{S}$, which is legal by Lemma A.1 in \cite{monaco2017adiabatic}, the proof comes to an end in the same manner as in Lemma~D.4 in \cite{henheikteufel20202}. 
\end{proof}

 \end{document}